\documentclass[amsart,11pt,oneside,english]{article}
\pdfoutput=1
\usepackage{amssymb,amsmath,amsfonts, amsmath,mathtools,amsthm,euscript,mathrsfs,MnSymbol,verbatim,enumerate,multirow,bbding,slashed,multicol,color,array, esint,babel, tikz,tikz-cd, tikz-3dplot,tkz-graph,pgfplots,ytableau,graphicx,float,rotating,
hyperref,
geometry,mathdots,savesym,cite, wasysym,amscd,graphicx,pifont,float,setspace,multirow,wrapfig,picture,subfigure, amsthm, enumerate}
\usepackage[normalem]{ulem} 
\usepackage[utf8]{inputenc}
\usepackage{prettyref}
\usepackage[T1]{fontenc}
\usepackage{datetime}
\usepackage{tabularx}

\numberwithin{equation}{section}
\usetikzlibrary{arrows,positioning,decorations.pathmorphing,   decorations.markings, matrix, patterns}
\hypersetup{colorlinks=true}
\hypersetup{linkcolor=black}
\hypersetup{citecolor=black}
\hypersetup{urlcolor=black}
\theoremstyle{plain}
\newtheorem*{thm*}{Theorem}
\theoremstyle{plain}% default
\newtheorem{thm}{Theorem}[section]

\newtheorem{lem}[thm]{Lemma}

\theoremstyle{definition}
\newtheorem{defn}[thm]{Definition}

\newtheorem{rem}[thm]{Remark}

\newcommand{\susu}{$\text{SU}(2)\!\times\!\text{SU}(3)$}
\newcommand{\sug}{$\text{SU}(2)\!\times\!\text{G}_2$}
\newcommand{\tr}{{\bf tr} }

\numberwithin{equation}{section}

\tikzset{
  big arrow/.style={
    decoration={markings,mark=at position 1 with {\arrow[scale=1.5,#1]{>}}},
    postaction={decorate},
    shorten >=0.4pt},
  big arrow/.default=black}
\begin{document}

\begin{titlepage}
\begin{center}
\vspace{2cm}
{\Huge\bfseries The Geometry of the \sug-model. \\  }
\vspace{2cm}
{\Large
Mboyo Esole$^{\diamond}$ and Monica Jinwoo Kang$^\clubsuit$\\}
\vspace{.6cm}
{\large $^{\diamond}$ Department of Mathematics, Northeastern University}\par
{\large  360 Huttington Avenue, Boston, MA 02115, U.S.A.}\par
\vspace{.3cm}
{\large $^\clubsuit$ Department of Physics,  Harvard University}\par
{ 17 Oxford Street, Cambridge, MA 02138, U.S.A}\par
\vspace{.3cm}
 \scalebox{.95}{\tt  j.esole@northeastern.edu,  jkang@fas.harvard.edu }\par
\vspace{2cm}
{ \bf{Abstract}}\\
\end{center}
We study  elliptic fibrations that geometrically engineer an \sug\  gauge theory realized  by a Weierstrass model for the  collision III+$\text{I}_0^{*\text{ns}}$. 
We find all the distinct crepant resolutions of such a model and the flops connecting them.
We compute the generating function for the Euler characteristic of the \sug-model. 
In the case of a Calabi-Yau threefold, we consider the compactification of M-theory and F-theory on an \sug-model to a five and six-dimensional supergravity theory with eight supercharges. 
By matching each crepant resolution with each Coulomb chamber of the five-dimensional theory, we determine the number of multiplets and compute the prepotential in each Coulomb chamber. 
In particular, we discuss the counting numbers of hypermultiplets in the presence of singularities.
We discuss in detail the cancellation of anomalies of the six-dimensional theory. 

\vspace{2cm}
\vfill 
\end{titlepage}

\tableofcontents
\newpage
\section{Introduction and Summary}\label{Sec:Intro}

Semi-simple Lie groups  appear naturally  in compactifications of M-theory and F-theory on elliptic fibrations \cite{Bershadsky:1996nh,Bershadsky:1996nu}. 
The Lie group is semi-simple  (but non-simple) when the discriminant of the elliptic fibration contains at least two  irreducible components $\Delta_1$ and $\Delta_2$ for which the dual graph of the singular fiber over the generic point of $\Delta_i$ $(i=1,2)$ is reducible.       These are called {\em collisions of singularities} and were first studied in string theory by Bershadsky and Johanson \cite{Bershadsky:1996nu} using Miranda's regularization of  singular Weierstrass models \cite{Miranda.smooth}.  
If the fiber over the generic point of $\Delta_i$ is a Kodaira fiber of type $T_i$, the collision is called a $T_1+T_2$-model. If $T_i$ has dual graph with Langlands dual $\widetilde{\mathfrak{g}}_i$, where $\mathfrak{g}_i$ is the Lie algebra of a compact simply connected group $G_i$ (assuming that the Mordell--Weil group is trivial), the model is also called a $G_1\times G_2$-model. 

The aim of this article  is to study the geometry and physics of  \sug-models realized by the  collision of singularities of type $\text{III}+\text{I}^{*\text{ns}}_0$:
\begin{equation}
\text{\sug\quad  from } \quad
\text{III}+\text{I}^{*\text{ns}}_0.\nonumber
\end{equation}
In F-theory, we associate to an elliptic fibration of dimension three or higher a group $G$ with Lie algebra $\mathfrak{g}$, and a 
  representation $\bf{R}$ of $\mathfrak{g}$.
 In the case of an \sug-model, the  representation $\bf{R}$ cannot be derived by the Katz--Vafa method (see Section \ref{sec:KatzVafa}) but can be deduced using the method of saturations of weights  obtained geometrically by intersection of fibral divisors and irreducible components of singular fibers over codimension-two points in the base (see Section \ref{sec:Sat}). The matter representation is then the  following  direct sum of  irreducible representations\footnote{We write the tensor product of a representation $\bf{r_1}$ of SU($2$) and $\bf{r_2}$ of G$_2$ as $(\bf{r_1}, \bf{r_2})$. We follow the usual tradition in physics of writing a representation by its dimension as there is no room for ambiguity with the representations used in this paper.  
The $\bf{2}$ and the $\bf{3}$ of SU($2$) are respectively the fundamental and the adjoint representation of SU($3$); the $\bf{7}$ and the $\bf{14}$ of G$_2$ are respectively the fundamental and the adjoint representation of G$_2$. The representations $(\bf{3,1})$,      $(\bf{1,14})$, and  $(\bf{1,7})$ are real while  the representation $(\bf{2,1})$ and  $(\bf{2,7})$ are pseudo-real. } (see Section \ref{sec:Rep})
\begin{equation}
\mathbf{R}=(\bold{3},\bold{1})\oplus(\bold{1},\bold{14})\oplus(\bold{2},\bold{7})\oplus (\bold{2},\bold{1})\oplus (\bold{1},\bold{7}).
\end{equation}
In contrast to the traditional approach \cite{Candelas:1997jz,Sadov:1996zm,Morrison:2012np}, we derive this representation $\bf{R}$ purely geometrically without using the anomaly cancellation conditions of the six-dimensional theory nor a dual heterotic model.  
This method is valid  for both  five and six-dimensional  theories and  has the advantage to work in situations where the   Katz--Vafa method is not applicable (see Section \ref{sec:KatzVafa}). 
In the $5d$ theory, it also has the benefit of not depending on the existence of a  $6d$ uplift. 
For the six-dimensional theory,  we check explicitly that the anomaly cancellation conditions are satisfied.  We also point out that the fundamental representation $(\bf{1,7})$ is a {\em frozen representation} when the curve supporting G$_2$ is a $-3$-curve intersecting the curve supporting SU($2$) transversally and at a unique point. This explains the absence of the representation $(\bf{1,7})$ in \cite{Candelas:1997jz}. 
The number of charged hypermultiplets in each of these representations are given in equation \eqref{eq1}, which generalizes  the spectrum $\tfrac{1}{2}(\bf{2,7})\oplus\tfrac{1}{2}(\bf{2,1})$ expected in a six-dimensional theory when a  \sug-model is locally defined at the collision of a $-2$-curve and a $-3$-curve intersecting transversally at one point and supporting respectively a SU($2$) and a G$_2$ Lie group \cite{Candelas:1997jz,Morrison:2012np}.

The $\text{III}+\text{I}^{*\text{ns}}_0$-model was first discussed in string theory in the early days of F-theory in \cite{Bershadsky:1996nu,Candelas:1996su,Candelas:1997jz} and was  already considered in mathematics in the early 1980's  by  Miranda \cite{Miranda.smooth}.   
  The \sug-model appears naturally in the study of {\em non-Higgsable clusters}  \cite{Morrison:2012np}.  Over non-compact bases, collisions of singularities are used to classify $6d$ ${\cal N}=(1,0)$ superconformal field theories using elliptic fibrations. 
 For example, for elliptic threefolds, such a non-Higgsable model is produced when the discriminant locus contains two rational curves with self-intersection $-3$ and $-2$ intersecting transversally  or three rational curves which form a chain of curves intersecting transversally at  a point with self-intersections $(-3,-2,-2)$ \cite{Bhardwaj:2015xxa,Heckman:2013pva,DelZotto:2014hpa,Morrison:2012np}. 
 
Surprisingly, despite  receiving a significant amount of attention in the last few years for their role in the study of superconformal field theories, many properties of the $\text{III}+\text{I}^{*\text{ns}}_0$ model remain unknown.
 For example, the structure of its Coulomb branch is not known. We show that the Coulomb branch consists of four chambers arranged as a Dynkin diagram of type A$_4$ illustrated on Figure \ref{fig:chambers}. We also compute the four corresponding prepotentials. 
 Geometrically, each Coulomb chamber corresponds to a different crepant resolution of the Weierstrass model. We present for the first time these  four crepant resolutions,  we study their fiber structures (see Figures \ref{fig:NKlist} and Tables \ref{fig:IIII0nsRes1}-\ref{Fig:I0Z3}) and determine the triple intersection numbers of their fibral divisors in Theorem \ref{Thm:TripleInt}.   
By comparing the triple intersection numbers with the prepotential, we determine   the number of charged hypermultiplets in a five-dimensional supergravity theory defined by a compactification of M-theory on a threefold given by an \sug-model.  We then show that these numbers are also consistent with an anomaly-free six-dimensional theory,  which corresponds to a compactification of F-theory on the same variety.

We work with an arbitrary base of dimension $n$ and specialize to the case of Calabi-Yau threefolds only when necessary to connect with the physics. 
Since the resolution  of singularities is a local process, the base can be either compact or non-compact. The compactness of the base will matter when considering anomaly cancellations.

Let $S$ and $T$ be the divisors supporting SU($2$) and G$_2$ respectively. 
Let $\Delta'$ be the third  component of the discriminant locus. Then, the number of charged hypermultiplets are (see Section \ref{sec:Physics})
\begin{align}\label{eq1}
\begin{aligned}
n_{\bf{2,7}} &= \frac{1}{2}S\cdot T, \quad &&  n_{\bf{3,1}} = g(S), &&  \quad  n_{\bf{2,1}} &&=-S\cdot(8 K +2S+ \frac{7}{2}T), \\
& \quad  
   &&   n_{\bf{1,14}}=g(T), && \quad n_{\bf{1,7}} &&= -T \cdot (5 K+2T+S).
 \end{aligned}
 \end{align}
 The numbers $n_{\bf{2,7}}$  and $n_{\bf{2,1}}$ can be half-integer numbers since  the representations  $(\bf{2,7})$  and $({\bf{2,1}})$ are both pseudo-real and thus can have   half-hypermultiplets charged under them.  
 As a sanity check, when $S$ and $T$ are respectively $-2$ and $-3$-curves intersecting transversally at a unique point, we retrieve the  familiar spectrum(see Remark \ref{rem1} below)
$$\bf{R}=\tfrac{1}{2}({\bf{2,7}})\oplus \tfrac{1}{2} ({\bf{2,1}}).$$   
 It is also useful to express $n_{\bf{2,1}}$ and $n_{\bf{1,7}}$ in terms of the genus\footnote{We recall that the genus $g$ of a curve $C$ in a surface of canonical class $K$ satisfies the relation  $2-2g=-C\cdot K-C^2$.}
 and the self-intersection of $T$ and $S$
  \begin{align}\label{eqn:gs}
 n_{\bf{2,1}}=16-16g(S) +6 S^2-\frac{7}{2} S\cdot T, 
 \quad  
 n_{\bf{1,7}} = 10-10g(T)+3 T^2 -S\cdot T.
 \end{align}
 The Hodge numbers of a compact Calabi-Yau threefold that is  an \sug-model are (see Section \ref{sec:Euler})
 \begin{equation}
 h^{1,1}(Y) =14-K^2, \quad   h^{2,1}(Y)=29 K^2+15 K \cdot S+24 K \cdot T+3 S^2+6 S \cdot T+6 T^2+14,
 \end{equation}
 where $K$ is the canonical class of the base of the elliptic fibration $Y$. 
In the six-dimensional supergravity theory, using Sadov's techniques \cite{Sadov:1996zm}, we check that anomalies are canceled explicitly by the Green--Schwarz--Sagnotti-West mechanism. The anomaly polynomial I$_8$ factors as a perfect square:
$$
I_8 =\frac{1}{2} \left(\frac{1}{2} K {\mathrm{tr}} R^2 +2 S {\mathrm{tr}}_{\bf{2}} F^2_1 + T {\mathrm{tr}}_{\bf{7}} F^2_2 \right)^2,
$$
where $F_1$ and $F_2$ are  the field strengths for SU($2$) and G$_2$ respectively.

\begin{rem}\label{rem1}
  the \sug-model, the fundamental representation $(\bold{1},\bold{7})$ is often ignored because it is localized  away from the collision of the curves supporting G$_2$ and SU($2$). However, it is expected from the study of the G$_2$-model as discussed in detail \cite{G2}. 
In the case of an \sug-model for which SU($2$) is supported on a curve $S$ and G$_2$ is supported on a $-3$-curve $T$, such that $S$ and $T$ intersect transversally at one point, we see using equation \eqref{eqn:gs} that the fundamental representation $(\bold{1},\bold{7})$ is a {\em frozen representation}: there are vertical curves carrying the weights of such a representation, whereas the number of hypermultiplets  $n_{\bold{1},\bold{7}}$ charged under that representation  vanishes. 
 If $S$ and $T$ are smooth rational curves, the adjoint representation is frozen ($ n_{\bf{3,1}}= n_{\bf{1,14}}=0$). 
If $T$ is and $S$ are respectively $-3$ and $-2$ curves intersecting transversally at a point, we retrieve  the famous spectrum 
$\tfrac{1}{2}({\bf{2,7}})\oplus \tfrac{1}{2} ({\bf{2,1}})$ first observed in \cite{Candelas:1997jz} by Candelas, Perevalov and Rajesh; and popularized in recent years with the study of non-Higgsable clusters  \cite{Morrison:2012np}.

\end{rem}
\begin{rem}\label{rem2}
In an ${\cal  N}=(1,0)$ six-dimensional theory, the cancellations of anomalies of an \sug-model show that the number of hypermultiplets charged in the representation $(\bold{2},\bold{1})$ receives a contribution $S\cdot T/2$, which suggests that there are curves carrying the weights of that representation at the intersection of $S$ and $T$.  
The formula for $ n_{\bf{2,1}}$ derived using anomaly cancellations of the six-dimensional theory is consistent with the triple intersection numbers of the \sug-model. However, we do not see any evidence of such a representation in the weights carried by curves composing the singular fibers over $S\cap T$. 
They carry instead the weights of the representation $(\bf{2,7})$ which contain as a subset the weights of the representation $(\bf{2,1})$. 
 We note that if $S$ and $\Delta'$ have a non-empty intersection, there are rational curves  carrying  the weight of the representation $(\bold{2},\bold{1})$ away from the intersection of $S$ and $T$, exactly over the points where the fiber III degenerates to a fiber of type IV. 
We discuss this further in Section \ref{sec:odd}.
\end{rem}

\subsection{Outlook}

In the rest of the introduction, we will motivate the model by an argument of simplicity, spell out the questions that we aim to answer, and summarize the key results of the paper for the convenience of the reader.  We discuss the definition of the model, explain the structure of its Coulomb branch (and the geometric derivation of its matter content), present the 
 non-Kodaira fibers produced by the four crepant resolutions of the \sug-model, and discuss aspects of compactifications of M-theory and F-theory on an \sug-model that is a Calabi-Yau threefold. We also summarize  the counting of its charged hypermultiplets, and the cancellations of anomalies of the six-dimensional theory.

The rest of the paper is a detailed development of these points and is structured as follows. In  Section \ref{sec:summary}, we collect our geometric  results. 
In  Section \ref{SU2G2Collision}, we introduce the model that we study in this paper, define its Weierstrass model, its  crepant resolutions, compute the Euler characteristic of the crepant resolutions and the triple intersection of the fibral divisors. 
In the case of a Calabi-Yau threefold, we also compute the Hodge numbers. We compute the adjacency graph of the hyperplane arrangement associated with an \sug-model, and finally match the structure of the hyperplane arrangement with the flopping curves of the crepant resolutions. 
In Section \ref{sec:Physics}, we study the consequences of our geometric results for the physics of F-theory and M-theory compactified on an \sug-model. 
We discuss the subtleties of counting the number of hypermultiplets in presence of singularities in  Section \ref{sec:odd}.

\subsection{The simplest collisions of singularities}

The  \sug-model  is an important model not only in F-theory and M-theory but also in birational geometry. Mathematically, it will  appear naturally as a key model of collision of singularities solely based on the simplicity of its fiber structure. 
It is  natural to organize elliptic fibrations describing collision of singularities by the rank of the associated Lie algebra derived from F-theory. Geometrically, the rank of $G$ counts the number of fibral divisors produced by a crepant resolution and relates  to the relative Picard number of the  elliptic fibration via the Shioda--Tate--Wazir theorem \cite[Corollary 4.1]{Wazir}.

If we organize the collisions of singularities by the rank of the associated Lie algebra,  the simplest collisions will correspond to the collision of two singular fibers with dual graphs $\widetilde{\text{A}}_1$ and the associated gauge group is either
$$\text{Spin($4$)=SU($2$)$\times$SU($2$) \quad or  \quad SO($4$)=SU($2$)$\times$SU($2$)$/(\mathbb{Z}/2\mathbb{Z})$},$$
when the Mordell--Weil group is trivial or  $\mathbb{Z}/2\mathbb{Z}$, respectively. These models are studied in detail in  \cite{SO4}. 
 
The next simplest  case has an associated Lie group of the type SU($2$)$\times$ $G$, where $G$ is a simple Lie group of rank two. 
There are three possibilities when $G$ is a compact simple and simply connected Lie group of rank  two: 
$$
\text{SU($2$)$\times$ SU($3$), \quad SU($2$)$\times$ Sp($4$), \quad \text{or}\quad SU($2$)$\times$ G$_2$.}
$$
The \susu-model is interesting for its connection to the non-Abelian sector of the Standard Model and is studied in \cite{SU2SU3}. 
The others  are  QCD-like theories obtained by replacing  SU($3$) by another simple and simply connected group of rank two. 
The  SU($2$)$\times$ Sp($4$)-model is studied in \cite{EKY2}. 
The group G$_2$ is the smallest  simply connected Lie group with a trivial center and all its representation are real.  
 The G$_2$-model is analyzed in detail in \cite{G2}.

 In Miranda's regularization, the collisions of singularities are organized by their values of the $j$-invariant as only Kodaira fibers sharing the same $j$-invariant are considered in Miranda's model \cite{Miranda.smooth}. 
The $\text{III}+\text{I}^{*\text{ns}}_0$ is a collision of fibers for which the   $j$-invariant is $1728$. The fiber over the generic point of the collision is a  non-Kodaira fiber composed of a chain of five rational curves intersecting transversally with multiplicities 1-2-3-2-1. This fiber is a contraction of a Kodaira fiber of type III$^*$ whose  dual graph is the  affine Dynkin diagram of type  $\widetilde{\text{E}}_7$.

\begin{table}[htb]
\begin{center}
\begin{tabular}{|c|c| }
\hline
Rank & $G$-Model \\
\hline 
2 & Spin($4$)  \\
2 &   SO($4$)\\
3 &  SU($2$)$\times$ SU($3$) \\
3 &   SU($2$)$\times$ Sp($4$) \\
3 &  SU($2$)$\times$ G$_2$\\
\hline 
\end{tabular}
\end{center}
\caption{Simplest collisions of singularities. We assume that $G$ is a semi-simple non-simply connected group of rank  $2$ or $3$. The Spin($4$) and SO($4$)-models are studied in \cite{SO4}. The
$G=$SU($2$)$\times$ SU($3$) is studied in \cite{SU2SU3}, the  $G=$SU($2$)$\times$ Sp($4$) is studied in \cite{EKY2}, and the $G=$SU($2$)$\times$ G$_2$ is the subject of the present paper. }
 \end{table}

\subsection{Canonical problems in  F/M-theory}\label{sec:FMProblem}
 A standard set of questions in  F-theory and M-theory compactifications on  an elliptically fibered Calabi-Yau threefold $Y$ are the following \cite{Esole:2017csj,ES, EY, Anderson:2017zfm,Bonetti:2011mw }:
  \begin{enumerate}[(i)]
  \item{\bf Birational geometry.}
  Does the singular Weierstrass model have a crepant resolution? How many crepant resolutions does it have? How are they flop-connected with each other? What is the graph of the flops?
  \item {\bf Fiber geometry.} 
What is the fiber structure of  each crepant resolution? What are the vertical curves composing the singular fibers? What are  their weights? What representation do they carry?

\item {\bf Coulomb branch and charged hypermultiplets.} 
What is the structure of the Coulomb branch of the five-dimensional  $\mathcal{N}=1$ theory with Lie group $\mathfrak{g}$ and representation $\mathbf{R}$ geometrically engineered by  an elliptic fibration $Y$? 
 How many hypermultiplets transform under each of the irreducible components  of $\mathbf{R}$? Can we completely fix the number of charged multiplets by comparing the triple intersection numbers and the prepotentials?
\item {\bf Anomaly cancellations and uplift.} Is the five-dimensional theory always compatible with an uplift to an anomaly free six-dimensional theory? 
 What are the conditions to ensure  cancellations of  anomalies of a six-dimensional  $\mathcal{N}=(1,0)$ supergravity obtained by compactification of F-theory on $Y$? 
 Can we fix the number of multiplets by the six-dimensional anomaly cancellation conditions?

\item {\bf Topological invariants.} 
What is the Euler characteristic of a crepant resolution? What are the Hodge numbers of $Y$? What are the triple intersection numbers in each Coulomb branch?  
\end{enumerate}
These questions are closely related to each other and 
 have been addressed recently for many geometries, such as the  G$_2$, Spin($7$), and Spin($8$)-models\cite{G2}, F$_4$-models \cite{F4}, SU($n$)-models \cite{ES,ESY1,ESY2,EY,Grimm:2011fx},  and also for (non-simply connected) semi-simple groups such as the SO($4$) and Spin($4$)-models \cite{SO4},  
 SU($2$)$\times$SU($4$),     (SU($2$)$\times$SU($4$))/$\mathbb{Z}_2$,    SU($2$)$\times$Sp($4$), and (SU($2$)$\times$Sp($4$))/$\mathbb{Z}_2$-models \cite{EKY2}. 
  
  We answer these canonical questions for the  \sug-model.
     We first do not restrict ourselves to  Calabi-Yau threefolds, but discuss the crepant resolutions and the Euler characteristic without fixing the base of the fibration in the spirit of \cite{AE1,AE2,EKY1,ESY1,ESY2}. 
The computation of topological invariants such as the triple intersection numbers and the Euler characteristic are streamlined by recent pushforward theorems \cite{Euler}. 

For recent works in F-theory and birational geometry in physics, see  for example \cite{
DelZotto:2018tcj,Lee:2018ihr, Collinucci:2018aho,Bies:2018uzw,Apruzzi:2018oge,
Anderson:2017rpr,
Klevers:2017aku
,Baume:2017hxm
,Jefferson:2018irk
},  see also the recent review on F-theory  \cite{Weigand:2018rez} and reference within.  We refer to \cite{Heckman:2018jxk} for a review of six-dimensional theories. 
\subsection{Defining the \sug-model}

Given a morphism $X\to B$ and an irreducible divisor $S$ of $B$, the generic fiber over $S$ is by definition the fiber over its generic point $\eta$. 
Such a fiber $X_\eta$ is  a scheme over the residue field $\kappa$ of $\eta$.  
The residue field $\kappa$ is not necessarily geometrically closed. Some components of  $X_\eta$ can be irreducible as a $\kappa$-scheme but will decompose further  after a field extension $\kappa\to \kappa'$. 
 We denote the cyclic group (resp. the group of permutations) of $n$ distinct elements by $\mathbb{Z}_n$ (resp.  $S_n$). 
In the case of a flat elliptic fibration, the Galois group of the minimal field extension that allows all the irreducible components of $X_\eta$ to be geometrically irreducible  is $\mathbb{Z}_2$ for all non-split Kodaira fibers with the exception of I$_0^*$ for which the Galois group could also be $S_3$ or $\mathbb{Z}_3$ \cite{G2}.

Kodaira fibers classify geometric generic fibers of an elliptic fibration (the fiber defined over the algebraic closure of the residue field).  
The generic fiber (defined over the residue field $\kappa$) is classified by the Kodaira type of the corresponding geometric generic fiber together with the Galois group of the minimal field extension necessary to make all irreducible components of the fiber geometrically irreducible. 
The Galois group is always $\mathbb{Z}_2$ unless in the case of the fiber I$_0^{*}$ where it can also  be $\mathbb{Z}_3$ or $S_3$. 
Thus,   there are two distinct fibers with dual graph $\widetilde{\text{G}}^\text{t}_2$  as the  Galois group of an irreducible cubic can be the symmetric group $S_3$ or the cyclic group $\mathbb{Z}_3$ \cite{G2}. 
When we specify the type of Galois group, we write  I$_0^{*\text{ns}}$ as  I$_0^{* S_3}$
 or  I$_0^{*\mathbb{Z}_3}$.  An  I$_0^{*\mathbb{Z}_3}$-model is very different from an  
  I$_0^{*\mathbb{S}_3}$-model already at the level of the fiber geometry as discussed in \cite{G2}  and reflected in the Tables of section \S \ref{sec:FibEn} of the present paper.  
   Elliptically fibered threefolds  corresponding to  an I$_0^{*\mathbb{Z}_3}$-model  or an I$_0^{S_3}$-model give the same gauge group, matter representations and number of charged hypermultiplets. The I$_0^{*\mathbb{S}_3}$-model is the more general one and behaves better with respect to resolutions of singularities. We refer to \cite{G2} for more information.  
  In the rest of the paper, when we write $\text{I}^{* \text{ns}}_0$ without further explanation, we always mean the generic $\text{I}^{*S_3}_0$.

  There are five different Kodaira fibers with dual graph $\widetilde{\text{A}}_1$ and thus producing an $\text{SU($2$)}$, namely I$_2^{\text{ns}}$,  I$_2^{\text{s}}$, III, IV$^{\text{ns}}$, and I$_3^{\text{ns}}$.   The \sug\ could be realized by  any of the following ten models

$$\text{
 $\text{I}^{\text{ns}}_2+\text{I}^{*\text{ns}}_0$,  $\text{I}^{\text{s}}_2+\text{I}^{*\text{ns}}_0$,    $\text{III}+\text{I}^{*\text{ns}}_0$, $\text{I}_3^{\text{ns}}+\text{I}^{*\text{ns}}_0$,  or $\text{IV}^{\text{ns}}+\text{I}^{*\text{ns}}_0$
},
$$
where $\text{I}^{*\text{ns}}_0$ could be either $\text{I}^{* \mathbb{Z}_3}_0$ or $\text{I}^{*S_3}_0$.   For example, the non-Higgsable models of type \sug\ studied in the literature are typically of the type  $\text{III}+\text{I}^{* \mathbb{S}_3}_0$.  

A Weierstrass model for the collision $\text{III}+\text{I}^{*S_3}_0$ is  \cite{Miranda.smooth}
\begin{equation}\label{Eq:W}
\text{III}+\text{I}^{*S_3}_0 :\quad y^2 z= x^3 + f s t^2 x z^2 + g s^2 t^3 z^3.
\end{equation}
The discriminant locus is composed of three irreducible components $S$, $T$, and $\Delta'$: 
\begin{equation}
\Delta=s^3 t^6 (4 f^3 + 27 g^2s ),
\end{equation}
where $S=V(s)$ and $T=V(t)$ are two smooth Cartier divisors supporting respectively the fiber of type III and of type I$_0^{* \text{ns}}$. We assume that $S$ and $T$ intersect transversally. The fiber over the generic point of the leftover discriminant $\Delta'=4 f^3 + 27 g^2s$ is a nodal curve (Kodaira type I$_1$). Following Tate's algorithm, the type of the decorated Kodaira fibers depends on the Galois group of the associated associated cubic polynomial 
\begin{equation}
P(q)= q^3 +f s q +  g s^2.
\end{equation} 
Assuming that $P(q)$ is irreducible, the Galois group is    $\mathbb{Z}_3$ if the discriminant of $P(q)$,  $\Delta(P)= s^3(4 f^3  +27 g^2 s)$,  is a perfect square in the residue field of the generic point of $T$ \cite{G2}. 
A simple way to have a $\mathbb{Z}_3$ Galois group is to increase the valuation of $f$ along $T$ \cite{G2}:
\begin{equation}\label{Eq:W2}
\text{III}+\text{I}^{*\mathbb{Z}_3}_0 :\quad y^2 z= x^3 + f s t^3 x z^2 + g s^2 t^3 z^3.
\end{equation}
 In this case, the $j$-invariant will be zero over the generic point of $T$ in contrast to the case of equation \eqref{Eq:W} where the $j$-invariant is $1728$ on both $S$ and $T$.  
The fact that the Galois group of $P(q)$ is $\mathbb{Z}_3$ is clear as it now takes the form
\begin{equation}
P(q)= q^3 +  g s^2.
\end{equation}

\subsection{Representations, Coulomb branches, hyperplane arrangements, and flops}
In F-theory, we associate to an elliptic fibration a group $G$ with Lie algebra $\mathfrak{g}$, and a 
  representation $\bf{R}$ of $\mathfrak{g}$.
 In the case of an \sug-model, the  representation $\bf{R}$ is  the direct sum of the following irreducible representations
  (see Section \ref{sec:SU2G2Res})
\begin{equation}
\mathbf{R}=(\bold{3},\bold{1})\oplus(\bold{1},\bold{14})\oplus(\bold{2},\bold{7})\oplus (\bold{2},\bold{1})\oplus (\bold{1},\bold{7}).
\end{equation}
The $(\bold{3},\bold{1})$ is the adjoint representation of SU($2$) and the $(\bold{1},\bold{14})$ is the adjoint representation of G$_2$. The $(\bold{2},\bold{7})$ is the bifundamental representation of \sug\  supported at the intersection $S\cap T$ of the two divisors supporting G$_2$ and SU($2$). 
The representation $(\bold{2},\bold{1})$  (resp.  $(\bold{1},\bold{7})$) is the fundamental representation of SU($2$)  (resp. G$_2$) supported at the collision of the third component of the discriminant locus  $\Delta'=(4 f^3 + 27 g^2 s)$ with the divisor $S$  (resp. $T$).

The study of the Coulomb branch of the gauge theory geometrically engineered by an elliptic fibration is the study of the minimal models over the Weierstrass model and how they  flop to each other. 
This can also be described 
through the hyperplane arrangement I$(\mathfrak{g},\mathbf{R})$  with hyperplanes 
that are kernels of the weights of $\mathbf{R}$ restricted inside the  dual fundamental Weyl chamber of $\mathfrak{g}$  \cite{ESY1,ESY2,EJJN1,EJJN2}. 
    In the case of the \sug-model, it is interesting to notice that  the generic SU($2$)-model and G$_2$-model do not have any flops \cite{ESY1,G2}. However, the \sug-model has the   bifundamental representation $(\bold{2},\bold{7})$, which contains several weights whose kernels are hyperplanes intersecting the open dual Weyl chamber of $\text{A}_1\oplus\mathfrak{g}_2$ and giving four chambers, whose incidence graph is a chain illustrated  in Figure  \ref{fig:chambers}. 
In this way, we see that the hyperplane arrangement I$(\mathfrak{g},\mathbf{R})$ does not care about the fundamental and adjoint representations, namely $(\bold{2},\bold{1})$, $(\bold{1},\bold{7})$, $(\bold{3},\bold{1})$, and $(\bold{1},\bold{14})$, since only the weights $\varpi$ of the bifundamental representation $(\bold{2},\bold{7})$ define the hyperplanes $\varpi^\bot$ intersecting the interior of the dual Weyl chamber of $\mathfrak{g}=\text{A}_1\oplus\mathfrak{g}_2$; hence, it is enough to study only I($\text{A}_1\oplus\mathfrak{g}_2,(\bold{2},\bold{7}))$. 
\begin{figure}[htb]
\begin{center}
{\begin{tikzpicture}[scale=0.25]
\coordinate (A1) at (90:10cm) {};
\coordinate (A2) at (210:10cm) {};
\coordinate (A3) at (330:10cm) {};
\coordinate (A4) at (0,-5cm) {};
\coordinate (A5) at (-5.768cm,0) {};
\coordinate (A6) at (5.772cm,0) {};
\coordinate (A7) at (-2.89cm, 5cm) {};
\coordinate (A8) at (2.89cm, 5cm) {};
\draw[thick] (A1)--(A2)--(A3)--(A1);
\draw[thick] (A1)--(A4);
\draw[thick] (A5)--(A4);
\draw[thick] (A6)--(A4);
%\draw[thick] (A5)--(A8);\draw[thick] (A6)--(A7);
\node at (-3cm, 2.1cm) {$2$};
\node at (3cm, 2.1cm) {$3$};
\node at (-5cm, -2.8cm) {$1$};
\node at (5.2cm, -2.8cm) {$4$};
\end{tikzpicture}}
 \quad\quad\quad
\scalebox{1}{
\begin{tikzpicture}
\node[draw,circle,thick,scale=1.5] (A1) at (0,0) {$1$};
 \node[draw,circle,thick,scale=1.5] (A2)at (2.5,0) {$2$};
 \node[draw,circle,thick,scale=1.5] (A3) at (5,0) {$3$};
 \node[draw,circle,thick,scale=1.5] (A4) at (7.5,0) {$4$};
 \draw[thick] (A1)--node[above] {$\varpi^{(\bf{2},\bf{7})}_{5}$}(A2)--node[above] {$\varpi^{(\bf{2},\bf{7})}_{6}$}(A3)--node[above] {$\varpi^{(\bf{2},\bf{7})}_{7}$}(A4);
\end{tikzpicture}}

\caption{The chamber structure of I$(\text{A}_1\oplus\mathfrak{g}_2, (\mathbf{2},\mathbf{7}))$ and its  adjacency graph. This also represents the structure of the extended K\"ahler cone of an \sug-model.  
Replacing $(\mathbf{2},\mathbf{7})$ with $\mathbf{R}=(\bold{3},\bold{1})\oplus(\bold{1},\bold{14})\oplus(\bold{2},\bold{7})\oplus (\bold{2},\bold{1})\oplus (\bold{1},\bold{7})$ does not change the 
adjacency graph since the adjoint and fundamental representations do not intersect the interior of the Weyl chamber of $\text{A}_1\oplus\mathfrak{g}_2$. 
The interior walls are given by the weights 
$ \varpi^{(\bf{2},\bf{7})}_5=(1;-2,1)$, $\varpi^{(\bf{2},\bf{7})}_6=(1;1,-1)$, and $\varpi^{(\bf{2},\bf{7})}_7=(1;-1,0)$. 
}
\label{fig:chambers} 
\end{center}
\end{figure}
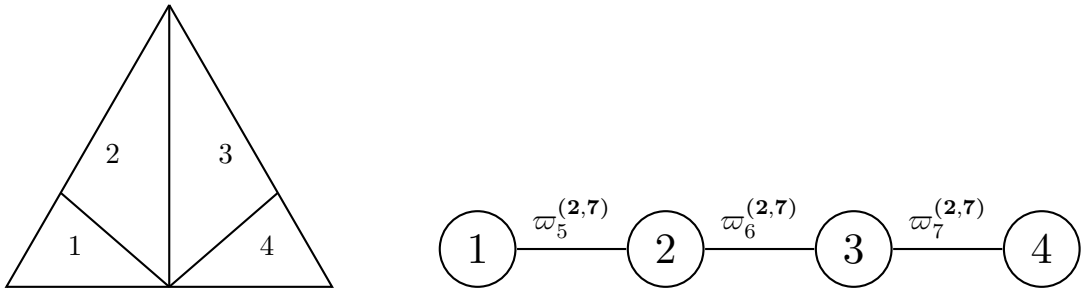

 \subsection{Non-Kodaira fibers}

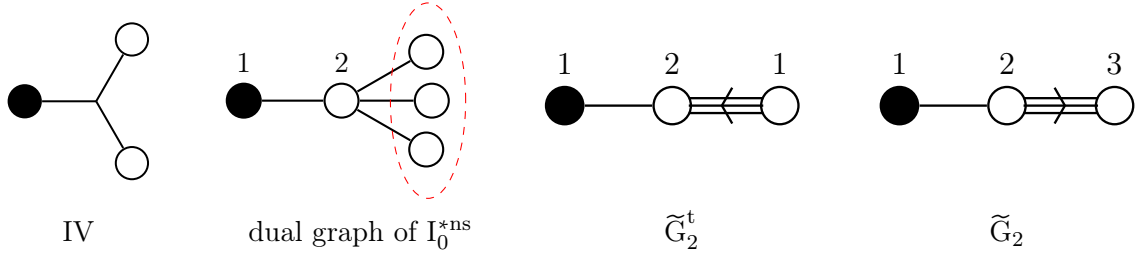
\begin{figure}[bth]
\begin{center}
\begin{tabular}{cccc}
%{\includegraphics[scale=.6]{IVGraph}} 
\raisebox{.8cm}{	\scalebox{.95}{$\begin{array}{c}
 \begin{tikzpicture}
				% G2
				%\node  at (1.3,1.5) {$\widetilde{\text{G}}_2$};
				\node[draw,circle,thick,scale=1.25, fill=black] (1) at (180:1){};
				\node[draw,circle,thick,scale=1.25] (2) at (-60:1){};
								\node[draw,circle,thick,scale=1.25] (3) at (60:1){};
												\draw[thick] (1) to (0,0);
				\draw[thick] (2) to (0,0);
				\draw[thick] (3) to (0,0);
			\end{tikzpicture}\end{array}
		$} }

 &		%\includegraphics[scale=.5]{G2dash}
 
 \raisebox{.8cm}{	\scalebox{1}{$\begin{array}{c}
 \begin{tikzpicture}
				% G2
				%\node  at (1.3,1.5) {$\widetilde{\text{G}}_2$};
				\node[draw,circle,thick,scale=1.25,label=above:{1}, fill=black] (1) at (0,0){};
				\node[draw,circle,thick,scale=1.25,label=above:{2}] (2) at (1.3,0){};
				\node[draw,circle,thick,scale=1.25] (3) at ($(2)+(30:1.3)$){};
								\node[draw,circle,thick,scale=1.25] (4) at ($(2)+(0:1.2)$){};
												\node[draw,circle,thick,scale=1.25] (5) at ($(2)+(-30:1.3)$){};
				\draw[thick] (1) to (2);
				\draw[thick] (2) to (3);
				\draw[thick] (2) to (4);
				\draw[thick] (2) to (5);
				\draw[red,dashed] (2.45,0) ellipse (.5cm and 1.3cm);
			\end{tikzpicture}\end{array}
		$} }

 &		\raisebox{1cm}{	\scalebox{1.1}{$\begin{array}{c}\begin{tikzpicture}
				% G2
				%\node  at (1.3,1.5) {$\widetilde{\text{G}}^\text{t}_2$};
				\node[draw,circle,thick,scale=1.25,label=above:{1}, fill=black] (1) at (0,0){};
				\node[draw,circle,thick,scale=1.25,label=above:{2}] (2) at (1.3,0){};
				\node[draw,circle,thick,scale=1.25,label=above:{1}] (3) at (2.6,0){};
				\draw[thick] (1) to (2);
				\draw[thick] (1.5,0.09) --++ (.9,0);
				\draw[thick] (1.5,-0.09) --++ (.9,0);
				\draw[thick] (1.5,0) --++ (.9,0);
				\draw[thick]
					(1.9,0) --++ (60:.25)
					(1.9,0) --++ (-60:.25);
			\end{tikzpicture}\end{array}
		$} }		& 

		\raisebox{1cm}{
		\scalebox{1.1}{$\begin{array}{c}\begin{tikzpicture}
				% G2
				%\node  at (1.3,1.5) {$\widetilde{\text{G}}_2$};
				\node[draw,circle,thick,scale=1.25,label=above:{1}, fill=black] (1) at (0,0){};
				\node[draw,circle,thick,scale=1.25,label=above:{2}] (2) at (1.3,0){};
				\node[draw,circle,thick,scale=1.25,label=above:{3}] (3) at (2.6,0){};
				\draw[thick] (1) to (2);
				\draw[thick] (1.5,0.09) --++ (.9,0);
				\draw[thick] (1.5,-0.09) --++ (.9,0);
				\draw[thick] (1.5,0) --++ (.9,0);
				\draw[thick]
					(2,0) --++ (120:.25)
					(2,0) --++ (-120:.25);
			\end{tikzpicture}\end{array}
		$} }\\
		IV&\quad dual graph of I$_0^{*\text{ns}}$ &\quad $\widetilde{\text{G}}^\text{t}_2$ \quad & $\widetilde{\text{G}}_2$ 
\end{tabular}
 \end{center}
		\caption{Conventions for dual graphs. {The black node represents the extra node of the affine Dynkin diagram.  The 
		affine Dynkin diagrams $\widetilde{\text{G}}_2$ and 
		$\widetilde{\text{G}}_2^\text{t}$ are Langlands dual of each other. 
		But only $\widetilde{\text{G}}_2^\text{t}$ is the  dual graph  of a singular fiber over the generic point of a component of the discriminant locusof an  elliptic fibrations. Specifically, $\widetilde{\text{G}}_2^\text{t}$ is the  dual graph of the fiber I$_0^{* \text{ns}}$.
	  } \label{fig:conv}}										
					
\end{figure}
Three rational curves that are  transverse  to each other and meet at the same point (such as the Kodaira fiber of type IV) are represented by three nodes connected to the same point. 
We write  $\widetilde{\text{G}}_2$ and $\widetilde{\text{G}}^\text{t}_2$ for the affine Dynkin diagram of type $\widetilde{\text{G}}_2$ and its Langlands dual, respectively. In Kac's notation \cite{Kac}, they are denoted respectively $\text{G}^{(1)}_2$ and
 ${\text{D}}^{(3)}_4$.  The dual graph of the fiber I$_0^{* \text{ns}}$ is of type $\widetilde{\text{G}}^\text{t}_2$ and not $\widetilde{\text{G}}_2$ as often stated in the F-theory literature but clear from Figure \ref{fig:conv}.

\begin{figure}[H]
\begin{center}
\scalebox{0.9}{

\begin{tabular}{cccc}
& Incomplete $\widetilde{\text{E}}_7$  & $\longrightarrow$ & Incomplete $\widetilde{\text{E}}_8$\\
Res. I&\begin{tabular}{c} \includegraphics[scale=.6]{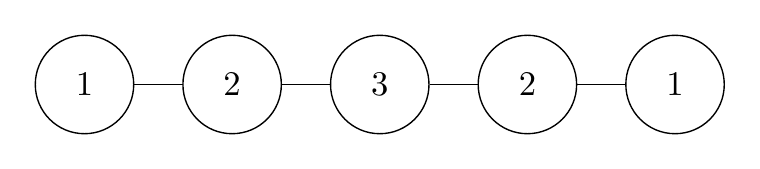} \end{tabular} 
 &  $\longrightarrow$  & 
  \begin{tabular}{c} \includegraphics[scale=.65]{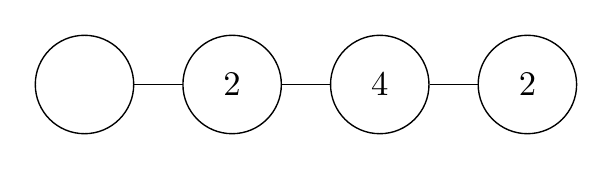} \end{tabular}   \\
Res. II &{\begin{tabular}{c} \includegraphics[scale=.6]{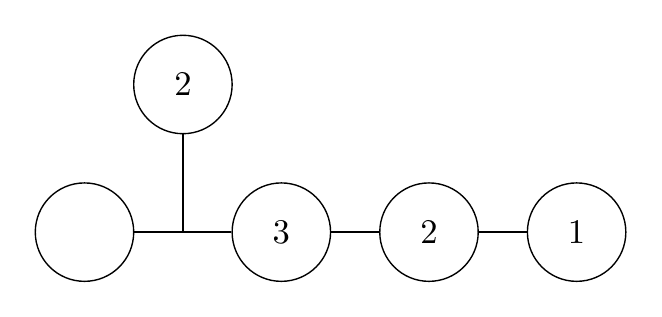} \end{tabular}}  & $\longrightarrow$& 
   \begin{tabular}{c} \includegraphics[scale=.6]{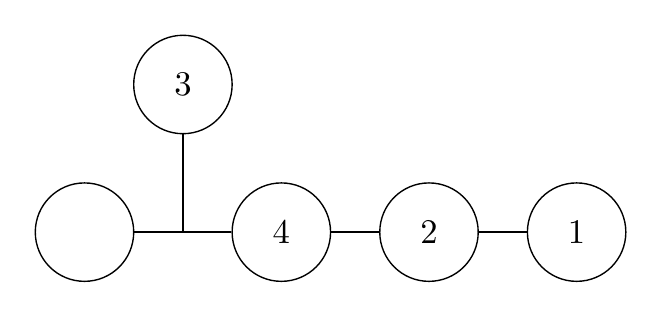} \end{tabular} \\
Res. III &\begin{tabular}{c} \includegraphics[scale=.6]{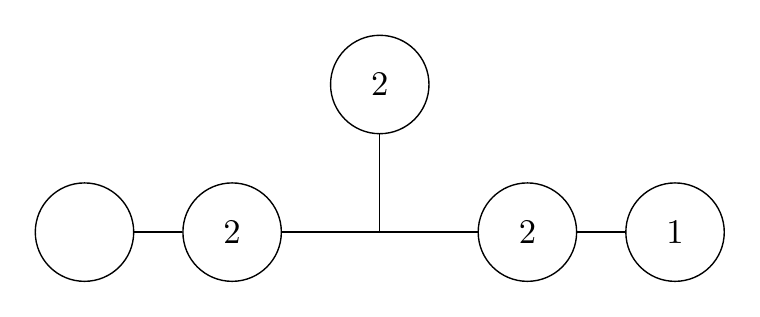} \end{tabular}  &$\longrightarrow$ & 
 \begin{tabular}{c} \includegraphics[scale=.6]{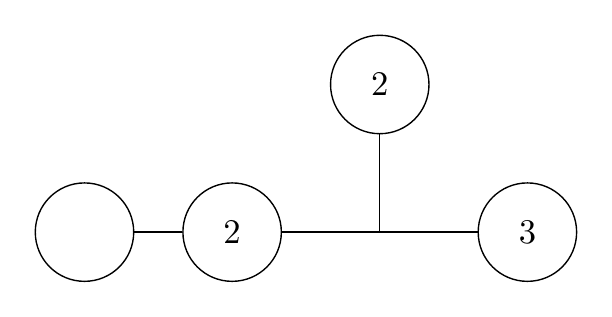} \end{tabular}\\
Res. IV  &  \begin{tabular}{c} \includegraphics[scale=.6]{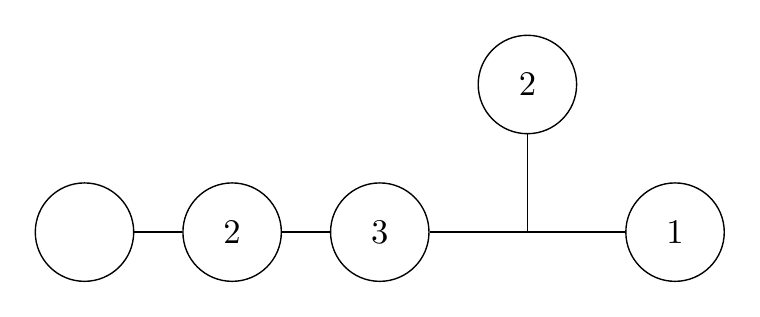} \end{tabular} &$\longrightarrow$ & 
  \begin{tabular}{c}\\  \\  \includegraphics[scale=.7]{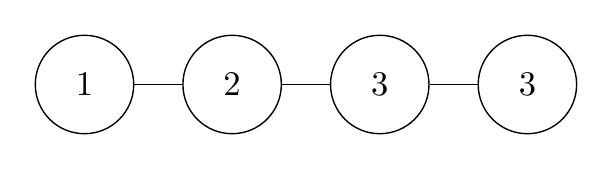} \end{tabular}  
  \end{tabular}}
\end{center}
\caption{Non-Kodaira Fibers for the \sug-model. On the left, we have the generic fiber over the collision  $S\cap T$. On the right, we have its specialization when $f=0$. This specialization is seen only when the base is at least a threefold or if we change the valuation of $t$ along $a_4$.  }
\label{fig:NKlist}
\end{figure}
 
 There are many examples of non-Kodaira fibers in the literature \cite{Miranda.smooth, Szydlo.Thesis,EY,EFY,Cattaneo,Hayashi:2014kca}.
  Over the generic point of $s=f=g=0$, the III-model has a non-Kodaira fiber  that are contractions of a fiber of type I$_0^*$. 
 The G$_2$-model has non-Kodaira fibers that are contractions of an I$_1^*$ and a IV$^*$ fiber \cite{G2}. 
At the collision III+I$_0^{*\text{ns}}$, we get a  non-Kodaira fiber that is an incomplete III$^*$ 
  and specializes further to an incomplete 
 II$^*$. 
 The fibers found by Miranda  at the collision III+I$_0^{*\text{ns}}$  match the ones we find in the resolution I for the generic fiber over $S\cap T$. 
Miranda has already noticed in \cite{Miranda.smooth} that the non-Kodaira fibers of Miranda's model were always contractions of Kodaira fibers. 
The same is true for Miranda's models of arbitrary dimension \cite{Szydlo.Thesis} and for flat elliptic threefolds \cite{Cattaneo}.

\subsection{Compactifications of F-theory and M-theory on an \sug-model}

The crepant resolutions of the Weierstrass model of an \sug-model are listed in equation \eqref{Eq:8Resolutions}. The Euler characteristic of an \sug-model obtained by one of these crepant resolutions is derived in Theorem  \ref{Thm:Eulerchar}.
As explained in the previous subsections, the hyperplane arrangements I($\mathfrak{g},\mathbf{R}$) has four  chambers whose adjacency graph is represented in Figure \ref{fig:chambers}. 
Each chamber corresponds to a specific crepant resolution that we determine explicitly. 
 While three of the crepant resolutions (namely, Resolutions I, III, and IV) are obtained by blowing up smooth centers, one (resolution II) is defined by a blowup with a non-smooth center. 
For each chamber, we match an  explicit crepant resolution of the Weierstrass model, so that the graph of flops matches the adjacency graph of the hyperplane arrangement.

We analyze the physics of compactifications of M-theory and F-theory on elliptically fibered Calabi-Yau threefolds corresponding to  \sug-models. 
These give five- and six-dimensional gauged supergravity  theories with eight supercharges with the gauge group \sug . We determine the matter content of these compactifications and study anomaly cancellations of the six-dimensional theory and their Chern-Simons terms. 

In the five-dimensional theory \cite{Cadavid:1995bk,IMS}, the structure of the Coulomb chambers is isomorphic to the adjacency graph of the hyperplane arrangement I($\mathfrak{g},\mathbf{R}$). 
We compute the one-loop prepotential in each  Coulomb chamber as a function of the number of hypermultiplets in each irreducible representations that add up to $\mathbf{R}$. 
  The Chern-Simons couplings are computed geometrically as triple intersection numbers of the fibral divisors (see Theorem \ref{Thm:TripleInt}). We match the triple intersection polynomial with the prepotential to obtain constraints on the number of charged hypermultiplets (see equation \eqref{eq:numbers}). 
In many cases, such a method will completely fix the number of multiplets, but here, the numbers $n_{\bf{2,1}}$ and $n_{\bf{3,1}}$ are left unfixed but related by a linear relation.  However, they are completely fixed once we use Witten's genus formula, which asserts that the number of multiplets in the adjoint representation is given by the genus of curve supporting the gauge group \cite{Witten:1996qb}.  

In the six-dimensional theory obtained by compactification of F-theory on an \sug-model, we solve the anomaly equations and deduce the number of hypermultiplets. They match perfectly what we found independently in the M-theory compactification (see equation \eqref{eq:numbersReal}).

\section{Geometric Results} \label{sec:summary}
In this section, we collect key geometric results on \sug-models. 

\subsection{Geometric description} \label{SU2G2Collision}

We consider the following  defining equation for an \sug-model:
\begin{eqnarray}\label{Eq:Weierstrass}
\begin{aligned}
\text{III}+\text{I}_0^{*\text{ns}} & : & y^2z =x^{3}+f st^{2}x z^2+g s^{2}t^{3} z^3.
\end{aligned}
\end{eqnarray}
 We assume that the coefficients $f$ and $g$  are 
algebraically independent and $S=V(s)$ and $T=V(t)$ are smooth divisors intersecting transversally. 
The Kodaira fiber over the generic point of  $S$   (resp. $T$) has a respective dual graph  $\widetilde{\text{A}}_1$ (resp. $\widetilde{\text{G}}^\text{t}_2$).
\subsection{Crepant resolutions}
\label{sec:crepres}
We use the following convention. 
 Let $X$ be a nonsingular variety. 
 Let $Z\subset X$ be a complete intersection defined by the transverse intersection of $r$ hypersurfaces $Z_i=V(g_i)$, where $g_i$ is a section of the line bundle $\mathscr{I}_i$ and $(g_1, \cdots, g_r)$ is a regular sequence. 
 We denote the blowup of a nonsingular variety $X$ along the complete intersection $Z$ by 
 $$\begin{tikzpicture}
	\node(X0) at (0,-.3){$X$};
	\node(X1) at (3,-.3){$\widetilde{X}.$};
	\draw[big arrow] (X1) -- node[above,midway]{$(g_1,\cdots ,g_{r}|e_1)$} (X0);	
	\end{tikzpicture}
	$$
The exceptional divisor is $E_1=V(e_1)$.	
 We abuse notation and use the same symbols for $x$, $y$, $s$, $e_i$ and their successive proper transforms. We also do not write the obvious pullbacks.

Each of the  following four sequences of blowups is a different crepant resolution of the   \sug-model  given by the Weierstrass model   of equation \eqref{Eq:Weierstrass}.

\begin{equation} \label{Eq:8Resolutions}
\begin{minipage}[c]{0.8\linewidth}
\begin{description}
\item[Resolution I:] 
\begin{tikzcd}[column sep=huge] X_0 \arrow[leftarrow]{r} {\displaystyle (x,y, s|e_1)} &  X_1 \arrow[leftarrow]{r} {\displaystyle (x,y,t| w_1)} &  X_2 \arrow[leftarrow]{r} {\displaystyle (y, w_1| w_2)} &  X_3  \end{tikzcd},
\item[Resolution II:]  
\begin{tikzcd}[column sep=huge] X_0 \arrow[leftarrow]{r} {\displaystyle (x,y, p_0|p_1)} &  X_1 \arrow[leftarrow]{r} {\displaystyle (y,p_1,t| w_1)} &  X_2 \arrow[leftarrow]{r} {\displaystyle (p_0,t| w_2)} &  X_3  \end{tikzcd}, 
\item[Resolution III:]
\begin{tikzcd}[column sep=huge] X_0 \arrow[leftarrow]{r} {\displaystyle (x,y,t|w_1)} &  X_1 \arrow[leftarrow]{r} {\displaystyle (x,y,s|e_1)} &  X_2 \arrow[leftarrow]{r} {\displaystyle (y, w_1| w_2)} &  X_3  \end{tikzcd}, 
\item[Resolution IV:] 
 \begin{tikzcd}[column sep=huge] X_0 \arrow[leftarrow]{r} {\displaystyle (x,y,t|w_1)} &  X_1 \arrow[leftarrow]{r} {\displaystyle (y,w_1| w_2)} &  X_2 \arrow[leftarrow]{r} {\displaystyle (x,y,s|e_1)} &  X_3  \end{tikzcd}.
\end{description}
\end{minipage}
\end{equation}
 These are embedded resolutions and  $X_0=\mathbb{P}(\mathscr{O}_B\oplus\mathscr{L}^{\otimes 2}\oplus\mathscr{L}^{\otimes 3})$.

\subsection{Intersection theory}
All of our intersection theory computations come down to  the following three theorems. 
The first is a theorem of Aluffi which gives the Chern class after a blowup along a local complete intersection. 
The second theorem is a pushforward theorem that provides a user-friendly method to compute invariants of the blowup space in terms of the original space. 
The last theorem is a direct consequence of functorial properties of the Segre class, and gives a simple method to pushforward analytic expressions in the Chow ring of a projective bundle to  the Chow ring of its base.

\begin{thm}[Aluffi, {
{\cite[Lemma 1.3]{Aluffi_CBU}}}]
\label{Thm:AluffiCBU}
Let $Z\subset X$ be the  complete intersection  of $d$ nonsingular hypersurfaces $Z_1$, \ldots, $Z_d$ meeting transversally in $X$.  Let  $f: \widetilde{X}\longrightarrow X$ be the blowup of $X$ centered at $Z$. We denote the exceptional divisor of $f$  by $E$. The total Chern class of $\widetilde{X}$ is then:
\begin{equation}
c( T{\widetilde{X}})=(1+E) \left(\prod_{i=1}^d  \frac{1+f^* Z_i-E}{1+ f^* Z_i}\right)  f^* c(TX).
\end{equation}
\end{thm}

\begin{thm}[Esole--Jefferson--Kang,  see  {\cite{Euler}}] \label{Thm:Push}
    Let the nonsingular variety $Z\subset X$ be a complete intersection of $d$ nonsingular hypersurfaces $Z_1$, \ldots, $Z_d$ meeting transversally in $X$. Let $E$ be the class of the exceptional divisor of the blowup $f:\widetilde{X}\longrightarrow X$ centered 
at $Z$.
 Let $\widetilde{Q}(t)=\sum_a f^* Q_a t^a$ be a formal power series with $Q_a\in A_*(X)$.
 We define the associated formal power series  ${Q}(t)=\sum_a Q_a t^a$, whose coefficients pullback to the coefficients of $\widetilde{Q}(t)$. 
 Then the pushforward $f_*\widetilde{Q}(E)$ is
 $$
  f_*  \widetilde{Q}(E) =  \sum_{\ell=1}^d {Q}(Z_\ell) M_\ell, \quad \text{where} \quad  M_\ell=\prod_{\substack{m=1\\
 m\neq \ell}}^d  \frac{Z_m}{ Z_m-Z_\ell }.
 $$ 
\end{thm}

\begin{thm}[{See  \cite{Euler} and  \cite{AE1,AE2,EKY1,Fullwood:SVW}}]\label{Thm:PushH}
Let $\mathscr{L}$ be a line bundle over a variety $B$ and $\pi: X_0=\mathbb{P}[\mathscr{O}_B\oplus\mathscr{L}^{\otimes 2} \oplus \mathscr{L}^{\otimes 3}]\longrightarrow B$ a projective bundle over $B$. 
 Let $\widetilde{Q}(t)=\sum_a \pi^* Q_a t^a$ be a formal power series in  $t$ such that $Q_a\in A_*(B)$. Define the auxiliary power series $Q(t)=\sum_a Q_a t^a$. 
Then 
$$
\pi_* \widetilde{Q}(H)=-2\left. \frac{{Q}(H)}{H^2}\right|_{H=-2L}+3\left. \frac{{Q}(H)}{H^2}\right|_{H=-3L}  +\frac{Q(0)}{6 L^2},
$$
 where  $L=c_1(\mathscr{L})$ and $H=c_1(\mathscr{O}_{X_0}(1))$ is the first Chern class of the dual of the tautological line bundle of  $ \pi:X_0=\mathbb{P}(\mathscr{O}_B \oplus\mathscr{L}^{\otimes 2} \oplus\mathscr{L}^{\otimes 3})\rightarrow B$.
\end{thm}

\subsection{Euler characteristics and Hodge numbers}\label{sec:Euler}
In the spirit of \cite{Euler}, the Euler characteristic depends only on the sequence of blowups. Using $p$-adic integration and the Weil conjecture, Batyrev proved the following theorem.

\begin{thm}[Batyrev, \cite{Batyrev.Betti}]
\label{thm:Batyrev}
Let $X$ and $Y$ be irreducible birational smooth $n$-dimensional projective algebraic varieties 
over $\mathbb{C}$. Assume that there exists a birational map $\varphi: X   - \rightarrow Y$ that does not 
change the canonical class. Then $X$ and $Y$ have the same Betti numbers. 
\end{thm}
 Batyrev's result was strongly inspired by string dualities, in particular by the work of Dixon, Harvey, Vafa, and Witten \cite{Dixon:1986jc}. 
As a direct consequence of Batyrev's theorem, the Euler characteristic of a crepant resolution of a variety with Gorenstein canonical singularities is independent on the choice of a crepant resolution. 
We identify the Euler characteristic as the degree  of the total  (homological) Chern class of a crepant resolution $f: \widetilde{Y }\longrightarrow Y$ of a Weierstrass model $Y\longrightarrow B$:
$$
\chi(\widetilde{Y})=\int c(\widetilde{Y}).
$$
We then use the birational invariance of the degree under the pushfoward to express the Euler characteristic as a class in the Chow ring of the projective bundle $X_0$. We subsequently push this class forward to the base to obtain a rational function depending only on the total Chern class of the base $c(B)$, the first Chern class $c_1(\mathscr L)$, and the class $S$ of the divisor in $B$:
$$
\chi(\widetilde{Y})=\int_B \pi_* f_* c(\widetilde{Y}).
$$
In view of Theorem \ref{thm:Batyrev}, this Euler characteristic is independent of the choice of a crepant resolution.

\begin{thm} \label{Thm:Eulerchar}
 The generating polynomial of the Euler characteristic of an \sug-model  obtained by a crepant resolution of a Weierstrass model given in Section \ref{sec:crepres}:
\begin{align}\nonumber
\chi (Y)=6 \ \frac{S^2-2 L-3 S L+2 (S^2-3 S L+S-2 L) T+(3 S+2) T^2}{(1+S) (1+T) (-1-6 L+2 S+3 T)} \ c(TB) .
\end{align}
\end{thm}
\begin{proof}
The total Chern class of $X_0=\mathbb{P}_B[\mathscr{O}_B\oplus \mathscr{L}^{\otimes 2} \oplus \mathscr{L}^{\otimes 3}]$ is 
$$
c(TX_0) = (1+H)(1+H+2\pi^*L)(1+H+3\pi^*L)\pi^* c(TB).
$$
where $L=c_1(\mathscr{L})$ and $H=c_1\Big(\mathscr{O}_{X_0}(1)\Big)$. The class of the Weierstrass equation is 
$[Y_0]=3H+6L$. 
Since all the resolutions are crepant, it is enough to do the computation in one of them. We consider Resolution I.  
We denote the blowups by 
$$f_1: X_1\to X_0, \quad  f_2:X_2\to X_1, \quad \text{and} \quad f_3: X_3\to X_2,$$ 
where $E_1$, $W_1$ and $W_2$ are respectively the classes of the first, second, and third blowups. 
  The center of the three blowups have  respectively classes: 
$$
\begin{array}{lll}
 Z_1^{(1)}=H+2\pi^* L, \quad &  Z_2^{(1)}= H+3\pi^*L, \quad  &Z_3^{(1)}=\pi^*S, \\
  Z_1^{(2)}=f_1^*H+2f_1^*\pi^*L-E_1,\quad & Z_2^{(2)}=f_1^* H+3f_1^* \pi^* L-E_1, \quad & Z_3^{(2)}=f_1^* \pi^* T, \\
   Z_1^{(3)}=f_2^*f_1^*H+3f_2^*f_1^*\pi^* L-f_2^*E_1-W_1, \quad & Z_2^{(3)}=W_1. & 
 \end{array}
$$
The successive blowups give (see Theorem \ref{Thm:AluffiCBU})
$$
\begin{aligned}
c(TX_1) &= \frac{(1+E_1)(1+Z_1^{(1)}-E_1) (1+Z_2^{(1)}-E_1) (1+Z_3^{(1)}-E_1)}{(1+Z_1^{(1)}) (1+Z_2^{(1)}) (1+Z_3^{(1)})} f_1^* c(T X_0),\\
c(TX_2) &= \frac{ (1+W_1)(1+Z_1^{(2)}-W_1) (1+Z_2^{(2)}-W_1) (1+Z_3^{(2)}-W_1)}{(1+Z_1^{(2)}) (1+Z 2^{(2)}) (1+Z 3^{(2)})} f_2^* c(T X_1),\\
c(T X_3) &= \frac{(1+W_2)(1+Z_1^{(3)}-W_2) (1+Z_2^{(3)}-W_2)}{(1+Z_1^{(3)}) (1+Z_2^{(1)})}f_3^* c(T X_2) .
\end{aligned}
$$
After the first blowup, the proper transform of $Y_0$ is of class  $Y_1=f_1^* Y_0-2E_1$.
After the second   blowup, the proper transform of $Y_1$ is of class  $Y_2=f_2^* Y_1-2W_1$.
And finally, after the third blowup, the proper transform of $Y_2$ is $Y=f_3^* Y_2-W_2$.
Altogether, we have 
$$
[Y]=(f_3^*f_2^*f_1^*(3H+6\pi^*L)-2 f_3^* f_2^* E_1-2f_3^*W_1-W_2)\cap [X_3].
$$
We also have that $c_1(X_3)= f_3^* f_2^* f_1^* c_1(X_0)-2 f_3^* f_2^* E_1-2f_3^*W_1-W_2$. Hence $c_1(Y)= f_3^* f_2^* f_1^*c_1(Y_0)$, which prove that the  resolution is crepant. 
The total Chern class of $Y$ is (see Theorem \ref{Thm:AluffiCBU})
$$
c(TY)\cap[Y]= 
\frac{ f_3^*f_2^*f_1^*(3H+6\pi^*L)-2 f_3^* f_2^* E_1-2f_3^*W_1-W_2}{1+f_3^*f_2^*f_1^*(3H+6\pi^*L)-2 f_3^* f_2^* E_1-2f_3^*W_1-W_2} c(TX_3)\cap [X_3].
$$
Then, 
$$
\chi(Y)= \int_{Y} c(TY)\cap[Y]= \int_B \pi_* f_{1*} f_{2^*} f_{3^*} c(TY)\cap[Y]
$$
The final formula for the Euler characteristic  follows directly from the pushforward  Theorems \ref{Thm:Push} and \ref{Thm:PushH}. 
\end{proof}
By direct expansion and specialization, we have the following three lemmas:
\begin{lem}
For an elliptic threefold, the Euler characteristic of  of an \sug-model  obtained by a crepant resolution of a Weierstrass model given in Section \ref{sec:crepres} is:
\begin{align}\nonumber
\chi (Y_3) =-6 (-2 c_1 \cdot L+12 L^2+S^2-5 S \cdot L+2 S \cdot T-8 L \cdot  T+2 T^2) .
\end{align}
\end{lem}
 By applying $c_1=L=-K$, we have the following Lemma.
\begin{lem}\label{lemEulerCY3}
In the case of a Calabi-Yau threefold,  The Euler characteristic 
 of an \sug-model  obtained by a crepant resolution of a Weierstrass model given in Section \ref{sec:crepres} is:
\begin{align}\nonumber
\chi (Y_3) =-6 (10 K^2+S^2+5 S \cdot  K+2 S \cdot T+8 K \cdot T+2 T^2) .
\end{align}
\end{lem}
\begin{lem}
The Euler characteristic for an elliptic fourfold, 
 the Euler characteristic 
 of an \sug-model  obtained by a crepant resolution of a Weierstrass model given in Section \ref{sec:crepres} 
  is given by
\begin{align}\nonumber
\begin{split}
\chi (Y_4)=-6 & \left( -2 c_2 \cdot L-72 L^3+12 c_1 \cdot L^2+ c_1 \cdot S^2-5 c_1 \cdot S \cdot  L+2 c_1 \cdot S \cdot T-8 c_1 \cdot L \cdot  T+2 c_1 \cdot T^2   \right. \\
& \ \ \left. +S^3-15 S^2 \cdot L +6 S^2 \cdot T+54 S \cdot L^2-44 S \cdot L \cdot T+9 S  \cdot T^2+84 L^2 \cdot T-34 L  \cdot T^2+4 T^3\right) .
\end{split}
\end{align}
\end{lem}
 Again, by the Calabi-Yau condition $c_1=L=-K$, we have the following Lemma.
\begin{lem}
In the case of a Calabi-Yau fourfold, the Euler characteristic 
 of an \sug-model  obtained by a crepant resolution of a Weierstrass model given in Section \ref{sec:crepres} is 
\begin{align}\nonumber
\begin{split}
\chi (Y_4)=-6 & \left(2 c_2  K+60 K^3+S^3+14 S^2  K+6 S^2  T+49 S  K^2+42 S K  T+9 S T^2  +76 K^2  T+32 K  T^2+4 T^3 \right) .
\end{split}
\end{align}
\end{lem}
\begin{thm}\label{Thm:Hodge}
In the Calabi-Yau case, the Hodge numbers of an \sug-model 
 obtained by the crepant resolution of a Weierstrass model given in Section \ref{sec:crepres} are
$$
h^{1,1}(Y)=14-K^2, \quad h^{2,1}(Y)=29 K^2+15 K S+24 K T+3 S^2+6 S T+6 T^2+14.
$$
\end{thm}
 There are three fibral divisors not touching the section  of the elliptic fibration. This number is exactly the rank of \sug . Hence, using the Shioda-Tata-Wazir theorem, we have 
$$
h^{1,1}(Y)=10+1+3-K^2, \quad h^{2,1}(Y)=h^{1,1}(Y)-\frac{1}{2}\chi(Y).
$$

\begin{thm}[{Shioda-Tate-Wazir; see Corollary 4.1. of \cite{Wazir}}]\label{Thm:STW}
Let $\varphi:Y\rightarrow B$ be a smooth elliptic threefold. Then, 
$$
\rho(Y)=\rho(B)+f+\mathrm{rank}(\text{MW} (\varphi))+1,
$$
where $f$ is the number of geometrically irreducible fibral divisors not touching the zero section.
\end{thm}
\begin{thm}\label{Thm:STW2} Let $Y$ be a smooth Calabi-Yau threefold elliptically fibered over a smooth variety $B$. 
Assuming  the Mordell-Weil group of $Y$ has rank zero, then 
\begin{equation}\nonumber
h^{1,1}(Y)=h^{1,1}(B)+f+1, \quad h^{2,1}(Y)=h^{1,1}(Y)-\frac{1}{2}\chi(Y),
\end{equation}
where $f$ is the number of geometrically irreducible fibral divisors not touching the zero section. In particular, if $Y$ is a $G$-model with $G$ a simple group, $f$ is the rank of $G$.
\end{thm}

\subsection{Triple intersection numbers} \label{sec:triple}
 Let $f:Y\to Y_0$ be one of the  crepant resolutions of a Weierstrass model given in Section \ref{sec:crepres}. 
Assuming that $Y$ is a threefold, we collect the triple intersection numbers $(D_a\cdot D_b \cdot D_c)\cap [Y]$ of the fibral divisors as the coefficient of  a cubic polynomial in $\psi_0$, $\psi_1$, $\phi_0$, $\phi_1$,  and $\phi_2$ that couples respectively with the fibral divisors $D_0^{\text{s}}$, $D_1^{\text{s}}$, $D_0^\text{t}$, $D_1^\text{t}$, and $D_2^\text{t}$. 
We pushforward to the Chow ring of $X_0$ and then to the base $B$. 
We recall that $\pi:X_0\to B$ is the projective bundle in which the Weierstrass model is defined. Then,
$$
\mathscr{F}_{trip}=
\int_Y\Big[\Big(
\psi_0 D_0^{\text{s}} +\psi_1 D_1^{\text{s}} +\phi_0 D_0^\text{t} +\phi_1 D_1^\text{t} +\phi_2 D_2^\text{t}  
\Big)^3\Big]=
\int_B\pi_* f_* \Big[\Big(
\psi_0 D_0^{\text{s}} +\psi_1 D_1^{\text{s}} +\phi_0 D_0^\text{t} +\phi_1 D_1^\text{t} +\phi_2 D_2^\text{t}  
\Big)^3\Big].
$$
Once the classes of the fibral divisors are determined, we  compute the pushforward of the triple intersection numbers using Theorem \ref{Thm:Push}  successively for each blowups $X_{i+1}\to X_i$  ($i=2,1,0$). We are then in the Chow ring of $X_0$ and we use Theorem \ref{Thm:PushH} to push forward to the Chow ring of $B$ using the projective bundle map $\pi: X_0\to B$. We can then take the degree to end up with a number. 

\begin{thm}\label{Thm:TripleInt}
The triple intersection polynomial of an \sug-model depends on the choice of a crepant resolution (they are listed in Section \ref{sec:crepres} ) and are as follows:
\quad
\begin{description}
\item[Resolution I:]
\begin{align}\nonumber
\begin{aligned}
\mathscr{F}^{\text{(I)}}_{trip}=& -9 T \phi _1 \phi _2 \left(\phi _2 (-6 L+S+3 T)-2 S \psi _1\right)+3 T \phi _1^2 \left(\phi _2 (-9 L+S+6 T)-2 S \psi _1\right) \\
& -2 \left(2 T \phi _2^3 (9 L-2 S-3 T)+S \left(\psi _0-\psi _1\right)^2 \left(2 \psi _0 (S-L)+\psi _1 (2 L+S)\right)+9 S T \psi _1 \phi _2^2\right) \\
&+3 T \phi _0 \left(\phi _1^2 (L-S)+2 S \phi _1 \left(\psi _1+\phi _2\right)-2 S \left(\left(\psi _0-\psi _1\right)^2+\phi _2^2\right)\right) \\
&+2 T \phi _0^3 (2 L-S-2 T)+3 T \phi _0^2 \left(\phi _1 (-2 L+S+T)-2 S \psi _1\right),
\end{aligned}
\end{align}
\item[Resolution  II:] 
\begin{align}\nonumber
\begin{aligned}
\mathscr{F}^{\text{(II)}}_{trip}=& S \psi _1^3 (-4 L-2 S+T)+T \phi _2^3 (-36 L+7 S+12 T) +4 T \phi _1^3 (L-T) -15 S T \psi _1 \phi _2^2 -3 S T\psi _1^2 \phi _2 \\
&+ 3 T(-9 L+S+6 T) \phi _1^2 \phi _2 -6ST \phi _1^2 \psi _1+9 T \phi _1 \phi _2 \left(2 S \psi _1 -(-6 L+S+3 T)\phi _2 \right)\\
& -T \phi _0^3 (-4 L+S+4 T) +3 T \phi _0^2 \left(\phi _1 (-2 L+S+T)-S \left(\psi _0+\psi _1+\phi _2\right)\right) \\
& -3 T \phi _0 \left(\phi _1^2 (S-L)-2 S \phi _1 \left(\psi _1+\phi _2\right)+S \left(-\psi _0+\psi _1+\phi _2\right)^2\right) +6ST \psi_0 \psi _1\phi _2 -3 S T \psi _0\phi _2^2  \\
& -3 S T\psi _0^2 \phi _2+3 S \psi _1 \psi _0^2 (-4 L+2 S+T)++3 S \psi _1^2 \psi _0 (4 L-T)-S \psi _0^3 (-4 L+4 S+T) ,
\end{aligned}
\end{align}
\item[Resolution III:]
\begin{align}\nonumber
\begin{aligned}
\mathscr{F}^{\text{(III)}}_{trip}=& -3 T(9 L-2 (S+3 T)) \phi _1^2 \phi _2 -3 T \phi _1 \phi _2^2 (-18 L+4 S+9 T) -3ST \psi _1 \phi _1^2 -12ST \psi _1\phi _2^2\\
& +T \phi _1^3 (4 L-S-4 T) -4ST \phi _2^3 (9 L-2 S-3 T) -3ST \psi_1^2 \phi _1 +12ST \psi _1\phi _1\phi _2 -2S(2 L+S-T)\psi_1^3 \\
& + 3 T \phi _0^2 \left(\phi _1 (T-2 L)-2 S \psi _0\right)+3 T \phi _0 \phi _1 \left(L \phi _1+2 S \psi _0\right) -3ST \psi _0 \phi _1^2 +6ST \psi _0 \phi _1\phi _2 \\
& +6ST \psi_0 \psi _1\phi _1  -3ST\psi _0^2 \phi _1 -2S(-2 L+2 S+T) \psi _0 \left(\psi _0^2 -2\psi_0^1 \psi _1+\psi_1^2\right) \\
& -2S(2 L+S-T) \psi _0 (\psi _0\psi _1 -2 \psi _1^2) -6ST\psi _0\phi _2^2 +4 T \phi _0^3 (L-T),
\end{aligned}
\end{align}
\item[Resolution IV:]
\begin{align}\nonumber
\begin{aligned}
\mathscr{F}^{\text{(IV)}}_{trip}=& S (-4 L-2 S+3 T) \psi _1^3 +4T (L-T)\phi _1^3 +9T(2 T-3 L) \phi _1^2 \phi _2  +27T(2 L-T) \phi _1 \phi _2^2 \\
& +12T(T-3 L) \phi _2^3 -6 S T \psi _1^2 \phi _2\\ 
& +S \psi _0^3 (4 L-4 S-3 T) +4 T \phi _0^3 (L-T)+\phi _0^2 \left(3 T \phi _1 (T-2 L)-6 S T \psi _0\right) \\
& +\psi _0^2 \left(3 S(-4 L+2 S+3 T) \psi _1 -6 S T \phi _2\right) +3 T \phi _0 \phi _1 \left(L \phi _1+2 S \psi _0\right) \\
& +3 S \psi _0 \left(\psi _1^2 (4 L-3 T)+4 T \psi _1 \phi _2 -2 T \phi _1^2 +6T \phi _1 \phi _2 -6T \phi _2^2\right). 
\end{aligned}
\end{align}
\end{description}
\end{thm}
\begin{proof}
We give the proof for the case of resolution I discussed in detail in Section  \ref{Sec:IIII0SResI}; the other cases follow the same pattern. 
$$
\begin{aligned}
\mathscr{F}_{trip} &=
\int_Y\Big[\Big(
\psi_0 D_0^{\text{s}} +\psi_1 D_1^{\text{s}} +\phi_0 D_0^\text{t} +\phi_1 D_1^\text{t} +\phi_2 D_2^\text{t}  
\Big)^3\Big]\\
&=\int_{X_3}\Big[\Big(
\psi_0 D_0^{\text{s}} +\psi_1 D_1^{\text{s}} +\phi_0 D_0^\text{t} +\phi_1 D_1^\text{t} +\phi_2 D_2^\text{t}  
\Big)^3(3H+6L-2E_1-2W_1-W_2)\Big] \\
&=\int_{X_0}f_{1*}f_{2*} f_{3*}\Big[\Big(
\psi_0 D_0^{\text{s}} +\psi_1 D_1^{\text{s}} +\phi_0 D_0^\text{t} +\phi_1 D_1^\text{t} +\phi_2 D_2^\text{t}  
\Big)^3(3H+6L-2E_1-2W_1-W_2)\Big] \\
&=\int_B\pi_* f_{1*}f_{2*} f_{3*} \Big[\Big(
\psi_0 D_0^{\text{s}} +\psi_1 D_1^{\text{s}} +\phi_0 D_0^\text{t} +\phi_1 D_1^\text{t} +\phi_2 D_2^\text{t}  
\Big)^3   (3H+6L-2E_1-2W_1-W_2)\Big].
\end{aligned}
$$
The classes of the fibral divisors in the Chow ring of $X_3$ are  
$$
[D_0^{\text{s}}]=S-E_1 , \quad [D_1^{\text{s}}]=E_1,\quad  [D_0^\text{t}]=T-W_1, \quad [D_1^\text{t}]=W_1-W_2, \quad  [D_2^\text{t}]=2W_2-W_1. 
$$
Denoting by $M$ an arbitrary divisor in the  class of the Chow ring of the base, The nonzero product intersection numbers of $M$, $H$, $E_1$, $W_1$, and $W_2$ are  
$$
\begin{aligned}
& \int_Y H^3=27L^2, \quad \int_Y E_1^3=-2 S (2 L+S), \quad \int_Y W_1^3=-2 T (2 L-S+T),    \quad \int_Y W_2^3=-T (5 L-2 S+T),\\
&  \int_Y W_1^2 E_1= -2 S T, \quad \int_Y W_1^2 W_2=T (-2 L+S-T), \quad \int_Y W_2^2 E_1=-2 S T, \quad \int_Y W_2^2 W_1= T (-L+S-2 T),\\
& \int_Y H^2 M=-9 L M, \quad \int_Y H M^2=3 M^2,\quad \int_Y E_1 W_1 W_2 =-ST, \\
& \int_Y MH^2=-9 L M, \quad \int_Y M E_1^2=-2S M,\quad \int_Y M W_1^2=-2T M,\quad \int_Y M W_2^2 =-2T M.
\end{aligned}
$$
The triple intersection numbers of the fibral divisors follow from these by simple linearity. 

\end{proof}
The triple intersection polynomials computed in Theorem \ref{Thm:TripleInt} are very different from each other in chambers I, II, III, and IV.

\subsection{Hyperplane arrangement} \label{sec:hyper}
We consider the  semi-simple Lie algebra $$\mathfrak{g}=\text{A}_1\oplus \mathfrak{g}_2.$$
 An irreducible representation of $\text{A}_1\oplus \mathfrak{g}_2$ is always a tensor product $\bf{R}_1\otimes \bf{R}_2$, where   $\bf{R}_1$ and $\bf{R}_2$ are respectively  irreducible representations of A$_1$  and $\mathfrak{g}_2$.  Following a common convention in physics, we denote a  representation  by its dimension in bold character.
The weights are denoted by $\varpi^I_j$ where the upper index I denotes the representation $\bf{R}_I$ and the lower index $j$ denotes a particular  weight of the representation $\bf{R}_I$. 
A weight of a representation of $\text{A}_1\oplus \mathfrak{g}_2$ is denoted by a triple $(a;b,c)$ such that $(a)$ is a weight of A$_1$ and $(b,c)$ is a weight of $\mathfrak{g}_2$, all in the basis of fundamental weights. 
 A vector $\phi$ of the coroot space of $\text{A}_1 \oplus\mathfrak{g}_2$   is written as $\phi=(\psi_1; \phi_1, \phi_2)$, where $\psi_1$ is the projection of $\phi$ on the fundamental coroot of A$_1$ and $(\phi_1,\phi_2)$ are the coordinates of the projection of $\phi$ along the coroot of G$_2$ expressed in the basis of fundamental coroots.  
 Each weight $\varpi$ defines a linear form by the natural evaluation on coroot vectors. We recall that fundamental coroots are dual to fundamental weights.

 We attach to an \sug-model the representation 
\begin{equation}
\mathbf{R}=(\bold{3},\bold{1})\oplus (\bold{1},\bold{14}) \oplus(\bold{2},\bold{7})\oplus (\bold{2},\bold{1})\oplus (\bold{1},\bold{7}).
\end{equation}
 The weights $\varpi^I_j$ of each summand are  listed in Table \ref{Table:Weight}.

\begin{table}[H]
\begin{center}
$
\begin{array}{|c|rrr|}
\hline
\text{Representation}& \multicolumn{3}{c|}{\text{Weights}}\\
\hline
(\bold{2},\bold{1}) &\quad \varpi^{(\bf{2},\bf{1})}_1=(1;0,0)\  \   &\varpi^{(\bf{2},\bf{1})}_2=(-1;0,0)    \  \   & \\
\hline
& \quad \varpi^{(\bf{1},\bf{7})}_1=(0;1,0)\   \  & \varpi^{(\bf{1},\bf{7})}_2=(0;-1,1)\       \  &  \varpi^{(\bf{1},\bf{7})}_3=(0;2,-1) \\
(\bold{1},\bold{7}) & &  \varpi^{(\bf{1},\bf{7})}_4=(0;0,0)& \\
& \quad \varpi^{(\bf{1},\bf{7})}_5=(0;-2,1)\   \  & \varpi^{(\bf{1},\bf{7})}_6=(0;1,-1)\       \  &  \varpi^{(\bf{1},\bf{7})}_7=(0;-1,0) \\
\hline
& \quad \varpi^{(\bf{2},\bf{7})}_1=(1;1,0)\   \  & \varpi^{(\bf{2},\bf{7})}_2=(1;-1,1)\       \  &  \varpi^{(\bf{2},\bf{7})}_3=(1;2,-1) \\
& &  \varpi^{(\bf{2},\bf{7})}_4=(1;0,0)& \\
(\bold{2},\bold{7})& \quad \varpi^{(\bf{2},\bf{7})}_5=(1;-2,1)\   \  & \varpi^{(\bf{2},\bf{7})}_6=(1;1,-1)\       \  &  \varpi^{(\bf{2},\bf{7})}_7=(1;-1,0) \\
& \quad \varpi^{(\bf{2},\bf{7})}_8=(-1;1,0)\   \  & \varpi^{(\bf{2},\bf{7})}_9=(-1;-1,1)\       \  &  \varpi^{(\bf{2},\bf{7})}_{10}=(-1;2,-1) \\
& &  \varpi^{(\bf{2},\bf{7})}_{11}=(-1;0,0)& \\
& \quad \varpi^{(\bf{2},\bf{7})}_{12}=(-1;-2,1)\   \  & \varpi^{(\bf{2},\bf{7})}_{13}=(-1;1,-1)\       \  &  \varpi^{(\bf{2},\bf{7})}_{14}=(-1;-1,0) \\
\hline
\end{array}
$
\end{center}
\caption{Weights of the representations of $\text{A}_1\oplus \mathfrak{g}_2$. \label{Table:Weight}}
\end{table}

We would like to  study the arrangement of hyperplanes  perpendicular to the weights of the representation $\mathbf{R}$ inside the dual  fundamental Weyl chamber of $\text{A}_1 \oplus \mathfrak{g}_2$. 
The only representation that will contribute to the interior walls is the bifundamental $(\bold{2},\bold{7})$.

\begin{thm}
The hyperplane arrangement $\text{I}(\text{A}_1\oplus\mathfrak{g}_2,\mathbf{R})$ with $\mathbf{R}=(\bold{2},\bold{7})$ has four chambers whose sign vectors and whose adjacency graph is  
given in Figure  \ref{fig:chambers}. A choice of a sign vector is  
$(\varpi^{(\bf{2},\bf{7})}_5,   \varpi^{(\bf{2},\bf{7})}_6,  \varpi^{(\bf{2},\bf{7})}_7)$. With respect to it, the chambers are listed in Table \ref{Table:Phases}. 
\end{thm}
\begin{proof}
The open dual fundamental Weyl chamber is the half cone defined by the  the positivity of the linear form induced by the simple roots: 
\begin{equation}
\psi_1>0, \quad 2\phi_1-\phi_2>0, \quad -3\phi_1+2\phi_2>0. 
\end{equation}
There are  only three  hyperplanes intersecting  the interior of the  fundamental Weyl chamber: 
$$
\varpi^{(\bf{2},\bf{7})}_5, \     \varpi^{(\bf{2},\bf{7})}_6,  \   \varpi^{(\bf{2},\bf{7})}_7.
$$
We use them in the order $(\varpi^{(\bf{2},\bf{7})}_5,   \varpi^{(\bf{2},\bf{7})}_6,  \varpi^{(\bf{2},\bf{7})}_7)$, the sign vector is 
 $$(\psi_1-2\phi_1-\phi_2,\psi_1 +\phi_1 - \phi_2,\psi_1 -  \phi_2).$$
 First we consider when $\psi_1-\phi_1>0$. Then $\varpi_5^{\bf{2,7}}=\psi_1-\phi_1>0$, $\varpi_6^{\bf{2,7}}=\psi_1+\phi_1-c=(\psi_1-\phi_1)+(2\phi_1-c)>0$, and $\varpi_7^{\bf{2,7}}=\psi_1-2\phi_1+\phi_2=(\psi_1-\phi_1)+(2\phi_1-\phi_2)+(-3\phi_1+2\phi_2)=\varpi_6^{\bf{2,7}}+(-3\phi_1+2\phi_2)>0$.

When $\psi_1-\phi_1<0$, $\varpi_5^{\bf{2,7}}=\psi_1-\phi_1<0$. Then we have either $\varpi_6^{\bf{2,7}}=\psi_1+\phi_1-\phi_2=(\psi_1-\phi_1)+(2\phi_1-\phi_2)>0$, or $\varpi_6^{\bf{2,7}}=\psi_1+\phi_1-\phi_2=(\psi_1-\phi_1)+(2\phi_1-\phi_2)<0$. If $\varpi_6^{\bf{2,7}}>0$, it follows that $\varpi_7^{\bf{2,7}}=\varpi_6^{\bf{2,7}}+(-3\phi_1+2\phi_2)>0$. When $\varpi_6^{\bf{2,7}}=\psi_1+\phi_1-\phi_2=(\psi_1-\phi_1)+(2\phi_1-\phi_2)<0$, we can have both $\varpi_7^{\bf{2,7}}>0$ and $\varpi_7^{\bf{2,7}}<0$.
 See Table \ref{Table:Phases}. 
\end{proof}

\begin{table}[H]
\begin{center}
$
\begin{array}{|c|ccc||c|}
\hline
\text{ Subchambers} &   \varpi^{(\bf{2},\bf{7})}_5 &   \varpi^{(\bf{2},\bf{7})}_6  &  \varpi^{(\bf{2},\bf{7})}_7 &  \text{Explicit description}  \\
 \hline  
\textcircled{1} & + & + & + &   
0<\frac{1}{2} \phi_2<\phi_1<\frac{2}{3} \phi_2,\quad\quad  \phi_1<\psi_1
\\
 \hline  
\textcircled{2} & + & +& -  & 
0<\frac{1}{2}\phi_2 <\phi_1<\frac{2}{3} \phi_2,\quad \phi_2-\phi_1<\psi_1<\phi _1
\\
 \hline  
\textcircled{3} & + & -& - &  
0<\frac{3}{2} \phi_1<\phi_2<2 \phi_1,\quad 2 \phi_1-\phi_2<\psi_1<\phi_2-\phi_1
 \\
 \hline  
\textcircled{4} & - & - & - & 
0<\frac{3}{2} \phi_1<\phi_2<2 \phi_1,\quad 0<\psi_1<2 \phi_1-\phi_2
 \\
 \hline  
\end{array}
$
\end{center}
\caption{Chambers of the hyperplane arrangement I($\text{A}_1\oplus\text{G}_2, \mathbf{R})$ with $\mathbf{R}=(\bf{2},\bf{7})$.  
 The wall between the chamber $i$  and $i+1$ ($i=1,2,3$)  is given by the weight  $\varpi^{(\bf{2},\bf{7})}_{8-i}$. 
\label{Table:Phases}  }
\end{table}

\subsection{Flops}\label{sec:Flops}
In this section, we discuss the flops between the resolutions I, II, III, IV. 
We identify the flopping curves in Table \ref{flopcurves}. These are curves whose weights differ only by a sign.

\begin{table}[H]
\begin{center}
\begin{tabular}{|c c c c c c c|c|}
\hline 
\multicolumn{7}{|c|}{Flopping curves} & Weight \\
\hline
\hline
Resolution I: & $\eta_1^{0A}$ & $[1;-1,0]$ \, $(\bf{2,7})$ & $\leftrightarrow$ & Resolution II: & $\eta_{01}^2$ & $[-1;1,0]$ \, $\bf{(2,7)}$ & $\omega_7^{(\bf{2,7})}$ \\
\hline 
Resolution II: & $\eta_1^{02}$ & $[1;1,-1]$ \, $\bf{(2,7)}$ & $\leftrightarrow$ & Resolution III: & $\eta_0^{12}$ & $[-1;-1,1]$ \, $\bf{(2,7)}$ & $\omega_6^{(\bf{2,7})}$ \\
\hline 
Resolution III: & $\eta_1^{12}$ & $[1;-2,1]$ \, $\bf{(2,7)}$ & $\leftrightarrow$ & Resolution IV: & $\eta_0^{2B}$ & $[-1;2,-1]$ \, $\bf{(2,7)}$ & $\omega_5^{(\bf{2,7})}$ \\
\hline 
\end{tabular}
\caption{ Flopping curves between different crepant resolutions of the \sug-model. 
 \label{flopcurves}}
\end{center}
\end{table}

\section{The Crepant Resolutions and Fiber Structures \label{sec:SU2G2Res}} 

In this section, we study the fibral structure of the elliptic fibrations obtained by the  crepant resolutions of the \sug-model given by the  Weierstrass model
\begin{equation}
Y_0: \quad y^{2}=x^{3}+fs t^{2}x+gs^{2}t^{3}.
\end{equation}
The resolutions are given by the sequence of blowups listed in Section \ref{SU2G2Collision}. 
We analyze the fiber structure of each of these crepant resolutions and determine the weights of the rational curves produced by the degeneration over codimension-two points. 
These weights are important to determine the representation $\mathbf{R}$. 
We denote the fibral divisors over $S$  and $T$  by $D_a^\text{s}$ and $D_a^\text{t}$ respectively. Their generic fibers are respectively written as $C_a^\text{s}$ and $C_a^\text{t}$.  
We will focus on analyzing the collision III+I$_0^{*ns}$ as we already know the behavior of the III-model and the G$_2$-model. 
\subsection{Resolution I}\label{Sec:IIII0SResI}
Resolution I is given by the following sequence of  blowups:
\begin{equation} 
\begin{tikzcd}[column sep=huge] 
X_0 \arrow[leftarrow]{r} {\displaystyle (x,y, s|e_1)} & X_1 \arrow[leftarrow]{r} {\displaystyle (x,y, t| w_1)} & X_2 \arrow[leftarrow]{r} {\displaystyle (y, w_1| w_2)} & X_3  
\end{tikzcd}  
\end{equation}
The proper transform of the Weierstrass model is
\begin{equation}Y: \quad w_{2}y^2=w_{1}(e_{1}x^{3}+fst^{2}x+gs^{2}t^{3}).
\label{eq:PTofIIII0s}
\end{equation}
The projective coordinates are then given by
\begin{equation}
[e_{1}w_{1}w_{2}x\,;\, e_{1}w_{1}w_{2}^{2}y\,;\, z=1][w_{1}w_{2}x\,;\, w_{1}w_{2}^{2}y\,;\, s][x\,;\, w_{2}y\,;\, t][y\,;\, w_{1}].
\end{equation}
The fibral divisors are given by $se_{1}=0$ for type III and $tw_{1}w_{2}=0$ for type I$_0^{* ns}$:
\begin{align} \label{eq:divisors}
\text{III}:
\begin{cases}
D_{0}^{\text{s}} :&  s=w_{2}y^{2}-w_{1}e_{1}x^{3}=0, \\
D_{1}^{\text{s}} :& e_{1}=w_{2}y^2-st^{2}w_{1}(fx+    gst)=0.
\end{cases}
\quad 
\text{I}_0^{* ns}:
\begin{cases}
D_{0}^\text{t} :& t=w_{2}y^{2}-w_{1}e_{1}x^{3}=0, \\
2D_{1}^\text{t} :& w_{1}=w_{2}=0, \\
D_{2}^\text{t} :& w_{2}=e_{1}x^{3}+fst^{2}x+gs^{2}t^{3}=0.
\end{cases}
\end{align}

On the intersection of $S$ and $T$, we see the following curves:
\begin{equation}
\text{On} \   S\cap T:
\begin{cases}
\begin{array}{cl}
D_{0}^\text{s}\cap D_{0}^{t}\rightarrow & \eta_0^0: s=t=w_{2}y^{2}-w_{1}e_{1}x^{3}=0,\\
D_{1}^\text{s}\cap D_{0}^{t}\rightarrow &\eta_1^{02}: e_{1}=t=w_{2}=0,\quad \eta_1^{0A}:e_{1}=t=y=0,\\
D_{1}^\text{s}\cap D_{1}^{t}\rightarrow &\eta_1^{12}: e_{1}=w_{1}=w_{2}=0,\\
D_{1}^\text{s}\cap D_{2}^{t}\rightarrow &\eta_1^{02}: e_{1}=w_{2}=t=0, \quad \eta_1^2: e_{1}=w_{2}=fx+gst=0.
\end{array}
\end{cases}
\end{equation}
Hence we can deduce that the five fibral divisors split in the following way to produce the fiber in Figure \ref{fig:IIII0s.Res1.cd2}, which is a fiber of type IV$^*$ with contracted nodes.
\begin{equation}
\text{On} \   S\cap T:
\begin{cases}
& \begin{tikzcd}  C_0^{\text{s}} \arrow[rightarrow]{r}  & \eta_0^0,  \end{tikzcd} \\
& \begin{tikzcd}  C_1^{\text{s}} \arrow[rightarrow]{r}  & 2\eta_1^{0A}+3\eta_1^{02}+2\eta_1^{12}+\eta_1^2,  \end{tikzcd}\\
& \begin{tikzcd}  C_0^\text{t} \arrow[rightarrow]{r}  & \eta_0^0+\eta_1^{02}+2\eta_1^{0A},  \end{tikzcd}\\
& \begin{tikzcd}  C_1^\text{t} \arrow[rightarrow]{r}  & \eta_1^{12},  \end{tikzcd}\\
& \begin{tikzcd}  C_2^\text{t} \arrow[rightarrow]{r}  & 2\eta_1^{02}+\eta_1^2.  \end{tikzcd}\\
\end{cases}
\end{equation}
\begin{figure}[H]
\begin{center}
\includegraphics{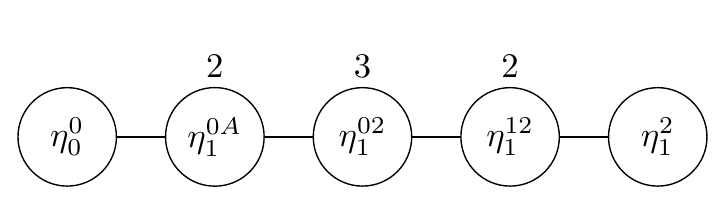}
\caption{Codimension-two Collision of \sug -model  at $S\cap T$, Resolution I. This fiber is of  type III$^*$ (with dual graph $\widetilde{\text{E}}_7$) with contracted nodes.} 
\label{fig:IIII0s.Res1.cd2}
\end{center}
\end{figure}
In order to get the weights of the curves, the intersection numbers are computed between the codimension-two curves and the fibral divisors.

\begin{table}[H]
\begin{center}
\begin{tabular}{|c|c|c|c|c|c|c|c|}
\hline 
  & $D^{\text{s}}_{0}$ & $D^{\text{s}}_{1}$ & $D^\text{t}_{0}$ & $D^\text{t}_{1}$ & $D^\text{t}_{2}$& Weight& Representation\\
\hline 
\hline 
$\eta_0^0$ & -2 & 2 & 0 & 0 & 0 & [-2;0,0] & $\bf{(3,1)}$\\
\hline 
$\eta_1^{0A}$ & 1 & -1 & -1 & 0 & 1 & [1;-1,0] & $\bf{(2,7)}$\\
\hline 
$\eta_1^{02}$ & 0 & 0 & 0 & 1 & -2 & [0;2,-1] & $\bf{(1,7)}\subset \bf{(1,14)}$\\
\hline 
$\eta_1^{12}$ & 0 & 0 & 1 & -2 & 3 & [0;-3,2] & $\bf{(1,14)}$\\
\hline 
$\eta_1^2$ & 0 & 0 & 0 & 1 & -2 & [0;2,-1] & $\bf{(1,7)}\subset \bf{(1,14)}$\\
\hline 
\end{tabular}
\end{center}
\caption{Weights and representations of the components of the generic curve over $S\cap T$  in Resolution I of the \sug -model. See Section \ref{sec:Sat} for more information on the interpretation of these representations.} 
\end{table}
%%%%%%%%%%%%%%%%%%%%%%%%%%%%%%%%%%%%%%%%

The fiber of Figure \ref{fig:IIII0s.Res1.cd2}  specializes further  when $f=0$: 
\begin{equation}
\begin{tikzcd}  \eta_1^2 \arrow[rightarrow]{r}  & \eta_1^{02}.  \end{tikzcd}
\end{equation}
This corresponds to the non-Kodaira diagram in Figure \ref{fig:IIII0s.Res1.cd3},  which is a fiber of type II$^*$ with contracted nodes.
\begin{figure}[H]
\begin{center}
\includegraphics{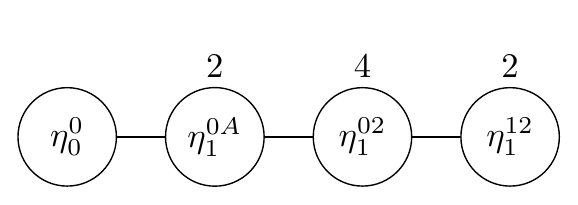}
\caption{Codimension-three enhancement of \sug -model at $S\cap T\cap V(f)$, Resolution I. This fiber is of type II$^*$ (with dual graph $\widetilde{\text{E}}_8$) with contracted nodes.}
\label{fig:IIII0s.Res1.cd3}
\end{center}
\end{figure}

\subsection{Resolution II}\label{Sec:IIII0SResII}
In this section, we study Resolution II in detail. Resolution II requires a first blowup that does not have a smooth center; it is useful to rewrite equation \eqref{Eq:Weierstrass} as
\begin{align}
Y_0 :
\begin{cases}
&y^2= x^{3}+fp_{0}t x+gp_{0}^{2}t , \\
&p_0=st.
\end{cases}
\end{align}
Resolution II is given by the following sequence of blowups
\begin{equation} 
\begin{tikzcd}[column sep=huge] 
X_0 \arrow[leftarrow]{r} {\displaystyle (x,y, p_0|p_1)} &  X_1 \arrow[leftarrow]{r} {\displaystyle (y,t,p_1| w_1)} &  X_2 \arrow[leftarrow]{r} {\displaystyle (t, p_0| w_2)} &  X_3  
\end{tikzcd} , 
\end{equation}
where $X_0=\mathbb{P}[\mathscr{O}_B\oplus\mathscr{L}^{\otimes 2}\oplus \mathscr{L}^{\otimes 3}]$. The projective coordinates are then 
\begin{equation}
[p_{1}w_1x\,:\, p_{1}w_1^{2}y\,:\, z=1][x\,:\, w_1y\,:\, p_0 w_2][y\,:\, tw_2\,:\, p_1][t\,:\, p_0],
\end{equation}
and the proper transform is 
\begin{align}
Y:
\begin{cases}
w_{1}y^2= p_{1}x^{3}+fp_{0}tw_2^{2}x+gp_{0}^{2}tw_{2}^{3},\\
p_0p_1=st.
\end{cases}
\end{align}
$X_1=Bl_{(x,y,p_0)} X_0$ has double point singularities at the ideal $(p_0,p_1, s,t)$.  
Recall that we have two curves from III and three curves from I$_0^{*}$ individually. We denote by $D^{\text{s}}_a$ and $D^\text{t}_a$ the fibral divisors that project to $S$ and $T$:
\begin{equation}
\begin{aligned}
\text{III}:
\begin{cases}
& D^{\text{s}}_0:  s=p_0=w_1y^2-p_{1}x^{3}=0 \\
 & D^{\text{s}}_1:  s=p_{1}=w_1y^2-(p_1x^3+p_0tw_2^2(fx+gp_0w_2)=0 
 \end{cases}
 \\
 \text{I}_0^{*}:
 \begin{cases}
& D^\text{t}_0:  w_2=p_0p_1-st=w_1y^2-p_1x^3=0 \\
 & D^\text{t}_1:  t=p_1=w_1=0 \\
 & D^\text{t}_2:  w_1=p_0p_1-st=p_{1}x^{3}+p_0tw_2^{2}(fx+gp_0w_2)=0
\end{cases}
\end{aligned}
\end{equation}
At the intersection of $S$ and $T$, the fiber enhances to a non-Kodaira fiber presented in Figure \ref{fig:IIII0s.Res2.cd2}, which is a fiber of type III$^*$ with contracted nodes. This is realized by the following splitting of curves. 
\begin{equation}
\text{On} \   S\cap T:
\begin{cases}
& \begin{tikzcd}  C^{\text{s}}_0 \arrow[rightarrow]{r}  & \eta_0^0+\eta_{01}^2  \end{tikzcd} \\
& \begin{tikzcd}  C^{\text{s}}_1 \arrow[rightarrow]{r}  & \eta_{01}^2 +3\eta_1^{02}+2\eta_1^{12}+\eta_1^2  \end{tikzcd}\\
& \begin{tikzcd}  C^\text{t}_0 \arrow[rightarrow]{r}  & \eta_0^0+\eta_1^{02}  \end{tikzcd}\\
& \begin{tikzcd}  C^\text{t}_1 \arrow[rightarrow]{r}  & \eta_1^{12}  \end{tikzcd}\\
& \begin{tikzcd}  C^\text{t}_2 \arrow[rightarrow]{r}  & 2\eta_{01}^2+2\eta_1^{02}+\eta_1^{12}+\eta_1^2  \end{tikzcd}\\
\end{cases}
\end{equation}
The curves at the intersection are given by
\begin{equation}
\text{On} \   S\cap T:
\begin{cases}
\begin{array}{cl}
D^{\text{s}}_{0}\cap D^\text{t}_{1}\rightarrow &\eta_0^0: s=p_{0}=w_2=w_1y^2-p_{1}x^{3}=0, \\
D^{\text{s}}_{0}\cap D^\text{t}_{2}\rightarrow&  \eta_{01}^{2}: s=p_{0}=w_1=p_1=0,\\
D^{\text{s}}_{1}\cap D^\text{t}_{0}\rightarrow &\eta_1^{02}:   s=p_{1}=w_2=w_1=0, \\
D^{\text{s}}_{1}\cap D^\text{t}_{1}\rightarrow & \eta_{1}^{12}: s=p_{1}=t=w_1=0, \\
D^{\text{s}}_{1}\cap D^\text{t}_{2}\rightarrow & \eta_1^2: s=p_{1}=w_1=fx+gp_0w_2=0,\quad \eta_{01}^2:  s=p_{1}=w_1=p_{0}=0, \\
				&\eta_{1}^{12}:  s=p_{1}=w_1=t=0, \quad\quad\quad\qquad  
				\eta_{1}^{02}:  s=p_{1}=w_1=w_2=0.\\
\end{array}
\end{cases}
\end{equation}

\vspace{-1cm}
\begin{figure}[H]
\begin{center}
\includegraphics{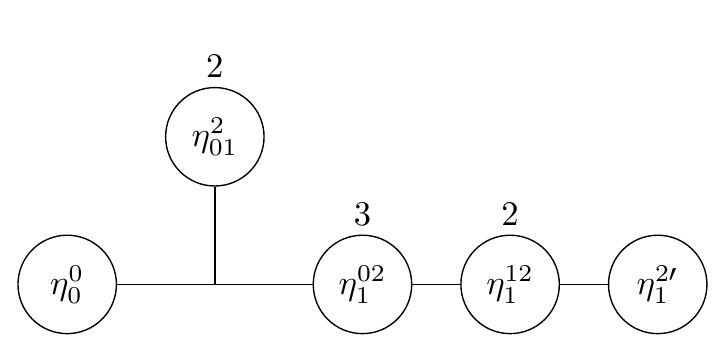}
\caption{Codimension-two Collision of \sug -model at $S\cap T$, Resolution II. This fiber is of type III$^*$ (with dual graph $\widetilde{\text{E}}_7$) with contracted nodes.} 
\label{fig:IIII0s.Res2.cd2}
\end{center}
\end{figure}

\begin{table}[htb]
\begin{center}
\begin{tabular}{|c|c|c|c|c|c|c|c|}
\hline 
& $D^{\text{s}}_{0}$ & $D^{\text{s}}_{1}$ & $D^\text{t}_{0}$ & $D^\text{t}_{1}$ & $D^\text{t}_{2}$& Weight& Representation\\
\hline 
\hline 
$\eta_0^0$ & -1 & 1 & -1 & 0 & 1 & [-1;-1,0] & $\bf{(2,7)}$\\
\hline 
$\eta_{01}^2$ & -1 & 1 & 1 & 0 & -1 & [-1;1,0] & $\bf{(2,7)}$\\
\hline 
$\eta_1^{02}$ & 1 & -1 & -1 & 1 & -1 & [1;1,-1] &   $\bf{(2,7)}$\\
\hline 
$\eta_1^{12}$ & 0  & 0 & 1 & -2 & 3 & [0;-3,2] &   $\bf{(1,14)}$\\
\hline 
$\eta_1^2$ & 0 & 0 & 0 & 1 & -2 & [0;2,-1] & $\bf{(1,7)}\subset \bf{(1,14)}$\\
\hline 
\end{tabular}
\end{center}
\caption{Weights and representations of the components of the generic curve over $S\cap T$  in Resolution II of the \sug -model. See Section \ref{sec:Sat} for more information on the interpretation of these representations.}
\end{table}

The fiber enhances further over $f=0$:
\begin{equation}
\begin{tikzcd}  \eta_1^2 \arrow[rightarrow]{r}  & \eta_{01}^2+\eta_1^{02}  \end{tikzcd} .
\end{equation}
For this codimension-three enhancement, we get a non-Kodaira fiber  corresponding to Figure \ref{fig:IIII0s.Res2.cd3}, which is a fiber of type II$^*$ with contracted nodes.
\begin{figure}[H]
\centering
\includegraphics{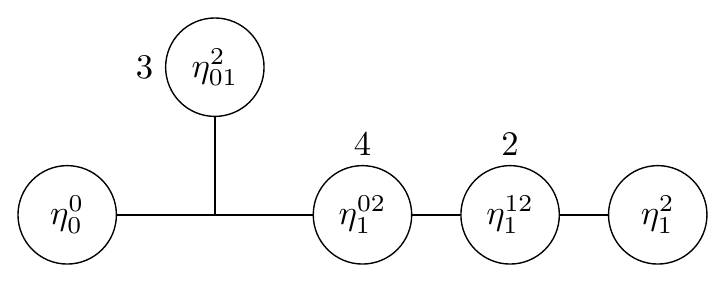}
\caption{Codimension-three enhancement of  \sug -model at $S\cap T\cap V(f)$, Resolution II. This fiber is of type II$^*$ (with dual graph $\widetilde{\text{E}}_8$) with contracted nodes.} 
\label{fig:IIII0s.Res2.cd3}
\end{figure}

\subsection{Resolution III}\label{Sec:IIII0SResIII}
Consider the following sequence of  blowups:
\begin{equation} 
\begin{tikzcd}[column sep=huge] 
X_0 \arrow[leftarrow]{r} {\displaystyle (x,y, t|w_1)} & X_1 \arrow[leftarrow]{r} {\displaystyle (x,y, s| e_1)} & X_2 \arrow[leftarrow]{r} {\displaystyle (y, w_1| w_2)} &  X_3 
\end{tikzcd}.
\end{equation}
The projective coordinates are then given by
\begin{equation}
[e_{1}w_{1}w_{2}x\,;\, e_{1}w_{1}w_{2}^{2}y\,;\, z=1][e_{1}x\,;\, e_{1}w_{2}y\,;\, t][x\,;\, w_{2}y\,;\, s][y\,;\, w_{1}].
\end{equation}
The proper transform is identical to equation \eqref{eq:PTofIIII0s}. It follows that the fibral  divisors are also identical to equation \eqref{eq:divisors}.

On the intersection of $S$ and $T$, we see the following curves:
\begin{equation}
\text{On} \   S\cap T:
\begin{cases}
\begin{array}{cl}
D_{0}^\text{s}\cap D_{0}^{\text{t}}\rightarrow &\eta_0^0 : s=t=w_{2}y^{2}-w_{1}e_{1}x^{3}=0, \\
D_{0}^\text{s}\cap D_{1}^{\text{t}}\rightarrow & \eta_0^{12}: s=w_{1}=w_{2}=0, \\
D_{0}^\text{s}\cap D_{2}^{\text{t}}\rightarrow &\eta_{01}^2: s=w_{2}=e_{1}=0, \quad \eta_0^{12}:  s=w_{2}=w_{1}=0 , \\
D_{1}^\text{s}\cap D_{1}^{\text{t}}\rightarrow &\eta_1^{12}: e_{1}=w_{1}=w_{2}=0, \\
D_{1}^\text{s}\cap D_{2}^{\text{t}}\rightarrow& \eta_1^2: e_{1}=w_{2}=fx+gst=0, \quad  \eta_{01}^2: e_{1}=w_{2}=s=0.
\end{array}
\end{cases}
\end{equation}
The resulting fiber is a non-Kodaira fiber illustrated in Figure \ref{fig:IIII0s.Res3.cd2} and corresponding to a fiber of type  III$^*$ with three contracted nodes.

\begin{equation}
\text{On} \   S\cap T:
\begin{cases}
& \begin{tikzcd}  C_0^{\text{s}} \arrow[rightarrow]{r}  & \eta_0^0+2\eta_0^{12}+\eta_{01}^2 , \end{tikzcd}\\
&\begin{tikzcd}  C_1^{\text{s}} \arrow[rightarrow]{r}  & \eta_{01}^2+2\eta_1^{12}+\eta_1^2 , \end{tikzcd}\\
& \begin{tikzcd}  C_0^\text{t} \arrow[rightarrow]{r}  & \eta_0^0,  \end{tikzcd}\\ 
 &\begin{tikzcd}  C_1^\text{t} \arrow[rightarrow]{r}  & \eta_0^{12}+\eta_1^{12},  \end{tikzcd}\\
 &\begin{tikzcd}  C_2^\text{t} \arrow[rightarrow]{r}  & 2\eta_{01}^2+\eta_1^2 . \end{tikzcd}\\
\end{cases}
\end{equation}
\begin{figure}[H]
\begin{center}
\includegraphics{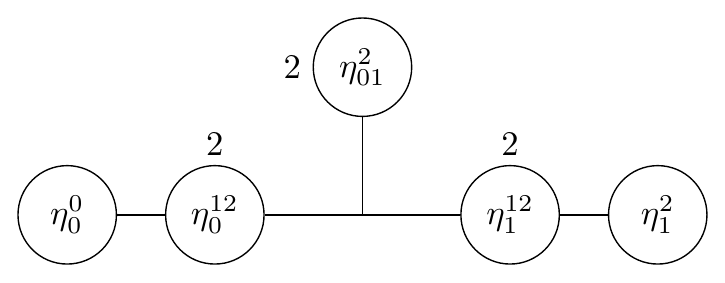}
\caption{Codimension-two Collision of \sug -model  at $S\cap T$, Resolution III. This fiber is of type III$^*$ (with dual graph $\widetilde{\text{E}}_7$) with contracted nodes.} \label{fig:IIII0s.Res3.cd2}
\end{center}
\end{figure}

In order to get the weights of the curves, the intersection numbers are computed between the codimension-two curves and the fibral divisors.
\begin{table}[H]
\begin{center}
\begin{tabular}{|c|c|c|c|c|c|c|c|}
\hline 
& $D^{\text{s}}_{0}$ & $D^{\text{s}}_{1}$ & $D^\text{t}_{0}$ & $D^\text{t}_{1}$ & $D^\text{t}_{2}$& Weight& Representation\\
\hline 
\hline 
$\eta_0^0$ & 0 & 0 & -2 & 1 & 0 & [0;0,-1] & $\bf{(1,7)}\subset \bf{(1,14)}$ \\
\hline 
$\eta_0^{12}$ & -1 & 1 & 1 & -1 & 1 & [-1;-1,1] & $\bf{(2,7)}$\\
\hline 
$\eta_1^{12}$ & 1 & -1 & 0 & -1 & 2 & [1;-2,1] & $\bf{(2,7)}$\\
\hline 
$\eta_1^2$ & 0 & 0 & 0 & 1 & -2 & [0;2,-1] & $\bf{(1,7)}\subset \bf{(1,14)}$\\
\hline 
$\eta_{01}^2$ & 0 & 0 & 0 & 1 & -2 & [0;2,-1] & $\bf{(1,7)} \subset \bf{(1,14)}$\\
\hline 
\end{tabular}
\end{center}
\caption{Weights and representations of the components of the generic curve over $S\cap T$  in Resolution III of the \sug -model.
See Section \ref{sec:Sat} for more information on the interpretation of these representations.
}
\end{table}
%%%%%%%%%%%%%%%%%%%%%%%%%%%%%%%%%
\vspace{-.5cm}
\begin{figure}[H]
\begin{center}
\includegraphics{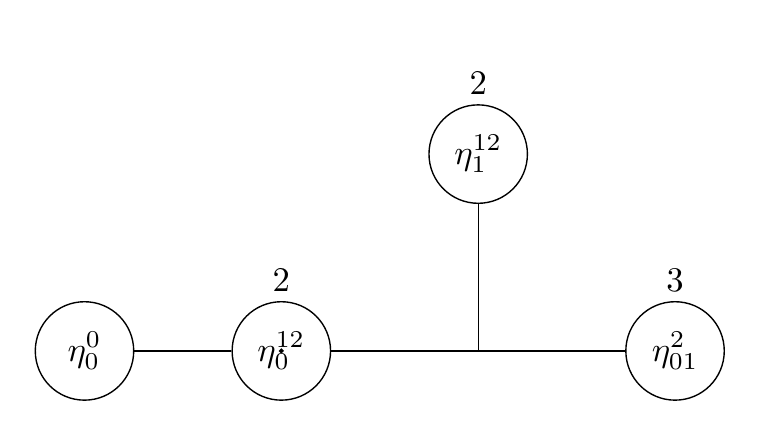}
\caption{ Codimension-three enhancement of \sug -model at $S\cap T\cap V(f)$, Resolution III. This fiber is of type II$^*$ (with dual graph $\widetilde{\text{E}}_8$) with contracted nodes.} \label{fig:IIII0s.Res3.cd3} 
\end{center}
\end{figure}

We have a further enhancement  when $f=0$ as  the rational curve $\eta_1^2$ coincides with $\eta_{01}^2$:  
\begin{equation}
\begin{tikzcd}  \eta_1^2 \arrow[rightarrow]{r}  & \eta_{01}^2. \end{tikzcd} 
\end{equation}
This corresponds to the codimension-three enhancement in Figure \ref{fig:IIII0s.Res3.cd3}, which is a fiber of type III$^*$ with contracted nodes.

\subsection{Resolution IV}\label{Sec:IIII0SResIV}
The last crepant resolution of the \sug-model is given by the following sequence of  blowups:
\begin{equation} 
\begin{tikzcd}[column sep=huge] 
X_0 \arrow[leftarrow]{r} {\displaystyle (x,y, t|w_1)} &  X_1 \arrow[leftarrow]{r} {\displaystyle (y, w_1| w_2)} & X_2 \arrow[leftarrow]{r} {\displaystyle (x,y, s| e_1)} &  X_3  .
\end{tikzcd}  
\end{equation}
Its projective coordinates are then given by
\begin{equation}
[e_{1}w_{1}w_{2}x\,;\, e_{1}w_{1}w_{2}^{2}y\,;\, z=1][e_{1}x\,;\, e_{1}w_{2}y\,;\, t][e_{1}y\,;\, w_{1}][x\,;\, y\,;\, s].
\end{equation}
The proper transform is identical to equation \eqref{eq:PTofIIII0s}. It follows that the divisors are also identical to equation \eqref{eq:divisors}.
On the intersection of $S$ and $T$, we see the following curves:
\begin{equation}\label{eq.21deg}
\text{On} \   S\cap T:
\begin{cases}
\begin{array}{cl}
D_0^{\text{s}}\cap D_0^\text{t}\rightarrow &  \eta_0^0: s=t=w_{2}y^{2}-w_{1}e_{1}x^{3}=0, \\
D_0^{\text{s}}\cap D_1^\text{t}\rightarrow &\eta_0^{12}: s=w_{1}=w_{2}=0,\\
D_0^{\text{s}}\cap D_2^\text{t}\rightarrow& \eta_{02}^{B}:s=w_{2}=x=0, \quad  \eta_{01}^{2} :s=w_{2}=e_{1}=0, \quad  \eta_0^{12}:s=w_{2}=w_{1}=0, \\
D_1^{\text{s}}\cap D_2^\text{t}\rightarrow & \eta_1^2: e_{1}=w_{2}=fx+gst=0 , \quad \eta_{01}^{2}: e_{1}=w_{2}=s=0.
\end{array}
\end{cases}
\end{equation}
Hence, we can deduce that the five fibral divisors split in the following way to produce the fiber in Figure \ref{fig:IIII0s.Res4.cd2}, which is a fiber of type III$^*$ with contracted nodes.
\begin{equation}
\text{On} \   S\cap T:
\begin{cases}
 & \begin{tikzcd}  C_0^{\text{s}} \arrow[rightarrow]{r}  & \eta_0^0+2\eta_0^{12}+3\eta_0^{2B}+\eta_{01}^2,  \end{tikzcd} \\
& \begin{tikzcd}  C_1^{\text{s}} \arrow[rightarrow]{r}  & \eta_{01}^2+\eta_1^2,  \end{tikzcd}\\
& \begin{tikzcd}  C_0^\text{t} \arrow[rightarrow]{r}  & \eta_0^0,\end{tikzcd} \\
& \begin{tikzcd} C_1^\text{t} \arrow[rightarrow]{r}  & \eta_0^{12},  \end{tikzcd}\\ 
%& \begin{tikzcd}    \end{tikzcd}\\ 
&\begin{tikzcd}  C_2^\text{t} \arrow[rightarrow]{r}  & 3\eta_{0}^{2B}+2\eta_{01}^2+\eta_1^2.  \end{tikzcd}\\
\end{cases}
\end{equation}

\begin{figure}[H]
\begin{center}
\includegraphics{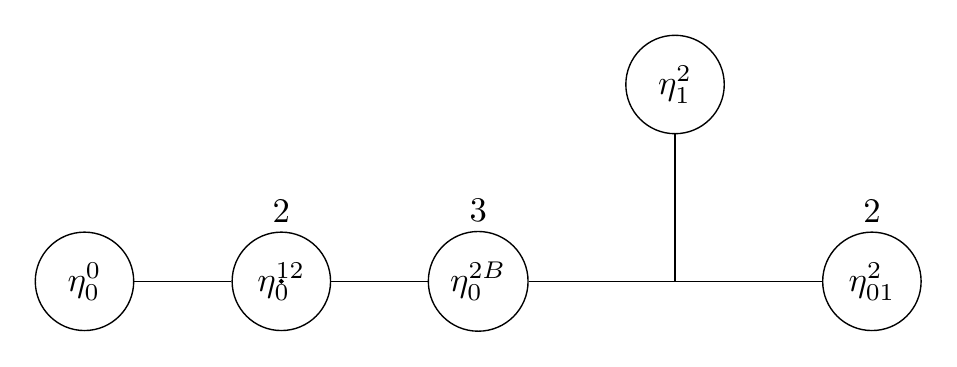}
\caption{Codimension-two Collision of \sug -model at $S\cap T$, Resolution IV. This fiber is of type III$^*$ (with dual graph $\widetilde{\text{E}}_7$) with contracted nodes.} \label{fig:IIII0s.Res4.cd2}
\end{center}
\end{figure}

In order to get the weights of the curves, the intersection numbers are computed between the codimension-two curves and the fibral divisors.
\begin{table}[htb]
\begin{center}
\begin{tabular}{|c|c|c|c|c|c|c|c|}
\hline 
  & $D^{\text{s}}_{0}$ & $D^{\text{s}}_{1}$ & $D^\text{t}_{0}$ & $D^\text{t}_{1}$ & $D^\text{t}_{2}$& Weight& Representation\\
\hline 
\hline 
$\eta_0^0$ & 0 & 0 & -2 & 1 & 0 & [0;0,-1] & $\bf{(1,7)}\subset \bf{(1,14)}$\\
\hline 
$\eta_0^{12}$ & 0 & 0 & 1 & -2 & 3 & [0;-3,2] & $\bf{(1,14)}$\\
\hline 
$\eta_{0}^{2B}$ & -1 & 1 & 0 & 1 & -2 & [-1;2,-1] & $\bf{(2,7)}$\\
\hline 
$\eta_{01}^{2}$ & 1 & -1 & 0 & 0 & 0 & [1;0,0] & $\bf{(2,1)}\subset \bf{(2,7)}$\\
\hline 
$\eta_1^2$ & 1 & -1 & 0 & 0 & 0  & [1;0,0] & $\bf{(2,1)}\subset \bf{(2,7)}$\\
\hline 
\end{tabular}
\end{center}
\caption{Weights and representations of the components of the generic curve over $S\cap T$  in Resolution IV of the \sug -model. 
See Section \ref{sec:Sat} for more information on the interpretation of these representations.\label{Table:Sat}
}
\end{table}
%%%%%%%%%%%%%%%%%%%%%%%%%%%%%%%%%%%%%

For the codimension-three fiber enhancement, consider when $f=0$.
Note that only the fiber $\eta_1^2$ changes under this condition:
\begin{equation}
\eta_1^2 \rightarrow \eta_{01}^2. 
\end{equation}
Even though only a single curve changed, we get a completely different fiber as a result. The codimension-three enhancement is represented in Figure \ref{fig:IIII0s.Res4.cd3}, which is a fiber of type III$^*$ with contracted nodes.
\begin{figure}[H]
\begin{center}
\scalebox{1.1}{\includegraphics{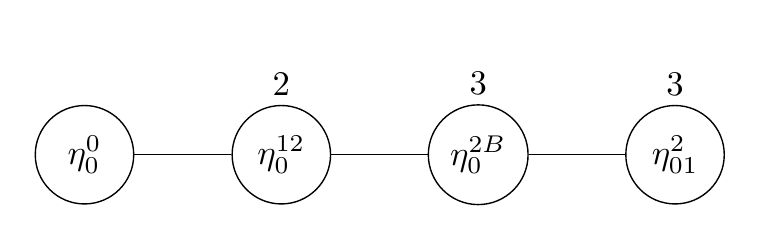}}
\caption{ Codimension-three enhancement of \sug -model at $S\cap T\cap V(f)$, Resolution IV. This fiber is of type II$^*$ (with dual graph $\widetilde{\text{E}}_8$) with contracted nodes.} \label{fig:IIII0s.Res4.cd3}
\end{center}
\end{figure}

\vspace{.5cm}

\subsection{Fiber structure}\label{sec:FibEn}
 In this section, we describe codimension-one and codimension-two fiber enhancements of all the resolutions (Resolutions I, II, III, IV) of the \sug-model. We consider two different ways of realizing the \sug-model: the collision of the fibers of type III+I$^{*S_3}_0$ and III + I$^{*\mathbb{Z}_3}_0$ depending on the Galois group of the fiber I$^{*\text{ns}}_0$. The Weierstrass equation of type III+I$^{*S_3}_0$ is given by $y^2=x^3+fs t^2 x+gs^2 t^3$, whereas that of type III + I$^{*\mathbb{Z}_3}_0$ is given by $y^2=x^3+ f s t^{3+\alpha} x + g s^2 t^3$. Resolutions I, II, III, and IV of the collision of the type III+I$^{*S_3}_0$ are respresented in Tables \ref{fig:IIII0nsRes1}, \ref{fig:IIII0nsRes2}, \ref{fig:IIII0nsRes3}, and \ref{fig:IIII0nsRes4} respectively. Resolutions I, II, III, and IV of the collision of type III + I$^{*\mathbb{Z}_3}_0$ are represented in Table \ref{Fig:I0Z3}.

%\clearpage
\begin{table}[H]
\begin{tikzcd}[column sep=normal]
& \raisebox{-1cm}{\includegraphics[scale=.6]{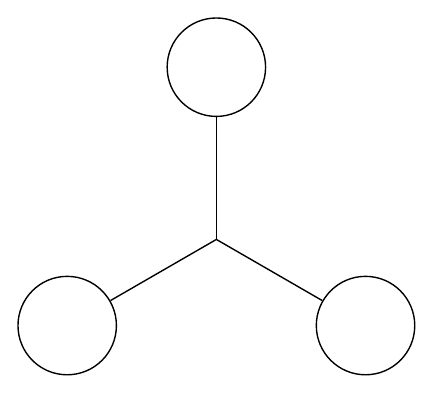}} \arrow[rightarrow]{r}{\displaystyle{g=0}} \arrow[rightarrow]{dd} {\displaystyle{T}} & \begin{tabular}{c} \includegraphics[scale=.8]{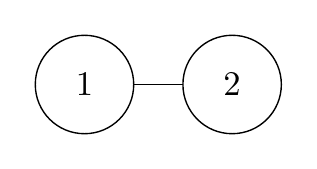} \end{tabular}  &  \\
\begin{tabular}{c} \includegraphics[trim=0cm  1cm 0 0cm,scale=1.4]{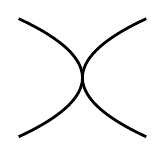} \end{tabular}  \arrow[rightarrow]{ru}{\displaystyle{f=0}} \arrow[rightarrow]{rd} {\displaystyle{T}}&   &    &  \\
 &\begin{tabular}{c} \includegraphics[scale=.8]{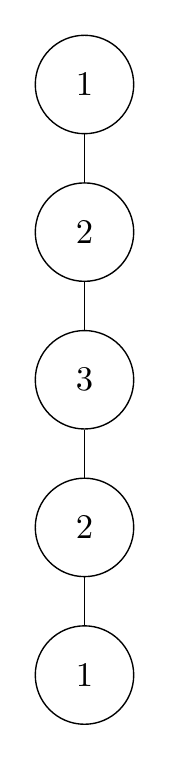} \end{tabular} \arrow[rightarrow]{r} {\displaystyle{f=0}}
 & 
  \begin{tabular}{c} \includegraphics[trim=0cm  0cm 0 0cm,scale=.8]{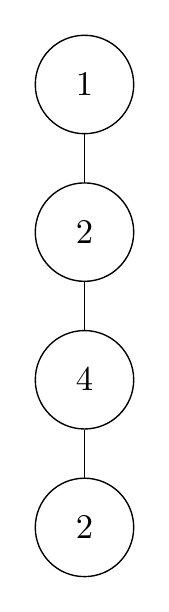} \end{tabular} &   \\
\begin{tabular}{c} \includegraphics[scale=.8]{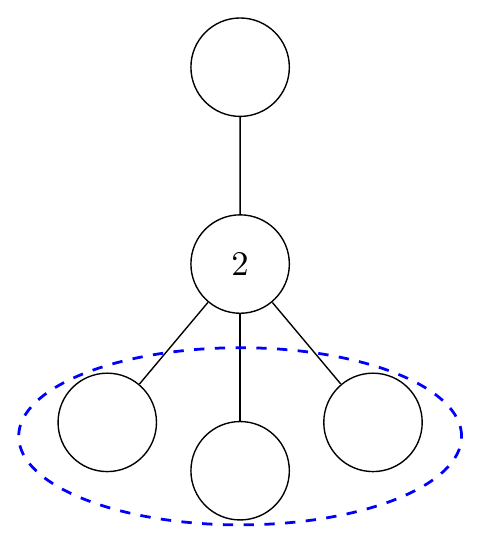}\end{tabular} \arrow[rightarrow]{ru}  {\displaystyle{S}}   \arrow[rightarrow]{r} {\displaystyle{4f^3+27 g^2 s=0}} &\raisebox{-1.5cm}{ \scalebox{.9}{ \includegraphics{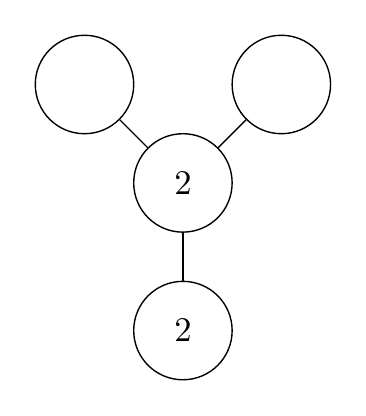}}} \arrow[rightarrow]{r} {\displaystyle{g=0}} \arrow[rightarrow]{ru}  {\displaystyle{S}} &     \raisebox{-1.5cm}{\includegraphics[scale=.9]{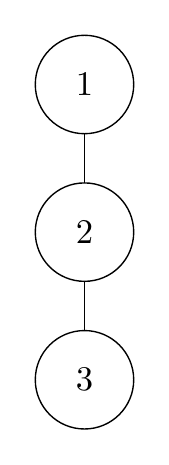}} & 
\end{tikzcd} 
\caption{ III + I$^{*S_3}_0$,  Resolution I. \label{fig:IIII0nsRes1}}
\end{table}

\newpage 
\begin{table}
\begin{tikzcd}[column sep=normal]
& \raisebox{-1cm}{\includegraphics[scale=.6]{IVGraph}} \arrow[rightarrow]{r}{\displaystyle{g=0}} \arrow[rightarrow]{dd} {\displaystyle{T}} & \begin{tabular}{c} \includegraphics[scale=.8]{H12} \end{tabular}  &  \\
\begin{tabular}{c} \includegraphics[trim=0cm  1cm 0 0cm,scale=1.4]{III} \end{tabular}  \arrow[rightarrow]{ru}{\displaystyle{f=0}} \arrow[rightarrow]{rd} {\displaystyle{T}}&   &    &  \\
 &{\begin{tabular}{c} \includegraphics[scale=.7]{Tstar1-2-321} \end{tabular}} \arrow[rightarrow]{r} {\displaystyle{f=0}}&  \begin{tabular}{c} \includegraphics[trim=0cm  0cm 0 0cm,scale=.7]{Tstar1-3-421} \end{tabular} &   \\
\begin{tabular}{c} \includegraphics[scale=.8]{G2dash}\end{tabular} \arrow[rightarrow]{ru}  {\displaystyle{S}}   \arrow[rightarrow]{r} {\displaystyle{4f^3+27 g^2 s=0}} &\raisebox{-1.5cm}{ \scalebox{.9}{ \includegraphics{TV22-1-1}}} \arrow[rightarrow]{r} {\displaystyle{g=0}} \arrow[rightarrow]{ru}  {\displaystyle{S}} &     \raisebox{-1.5cm}{\includegraphics[scale=.9]{TV123}} & 
\end{tikzcd} 
\caption{ III + I$^{*S_3}_0$,  Resolution II. \label{fig:IIII0nsRes2}}
\end{table}

\clearpage 
\newpage 
\begin{table}
\begin{tikzcd}[column sep=normal]
& \raisebox{-1cm}{\includegraphics[scale=.6]{IVGraph}} \arrow[rightarrow]{r}{\displaystyle{g=0}} \arrow[rightarrow]{dd} {\displaystyle{T}} & \begin{tabular}{c} \includegraphics[scale=.8]{H12} \end{tabular}  &  \\
\begin{tabular}{c} \includegraphics[trim=0cm  1cm 0 0cm,scale=1.4]{III} \end{tabular}  \arrow[rightarrow]{ru}{\displaystyle{f=0}} \arrow[rightarrow]{rd} {\displaystyle{T}}&   &    &  \\
 &\begin{tabular}{c} \includegraphics[scale=.7]{Tstar12-21-2} \end{tabular} \arrow[rightarrow]{r} {\displaystyle{f=0}}
 & 
  \begin{tabular}{c} \includegraphics[trim=0cm  0cm 0 0cm,scale=.7]{Tstar12-2-3} \end{tabular} &   \\
\begin{tabular}{c} \includegraphics[scale=.8]{G2dash}\end{tabular} \arrow[rightarrow]{ru}  {\displaystyle{S}}   \arrow[rightarrow]{r} {\displaystyle{4f^3+27 g^2 s=0}} &\raisebox{-1.5cm}{ \scalebox{.9}{ \includegraphics{TV22-1-1}}} \arrow[rightarrow]{r} {\displaystyle{g=0}} \arrow[rightarrow]{ru}  {\displaystyle{S}} &     \raisebox{-1.5cm}{\includegraphics[scale=.9]{TV123}} & 
\end{tikzcd} 
\caption{ III + I$^{*S_3}_0$,  Resolution III. \label{fig:IIII0nsRes3}}
\end{table}

\clearpage

 \newpage 
\begin{table}
\begin{tikzcd}[column sep=normal]
& \raisebox{-1cm}{\includegraphics[scale=.6]{IVGraph}} \arrow[rightarrow]{r}{\displaystyle{g=0}} \arrow[rightarrow]{dd} {\displaystyle{T}} & \begin{tabular}{c} \includegraphics[scale=.8]{H12} \end{tabular}  &  \\
\begin{tabular}{c} \includegraphics[trim=0cm  1cm 0 0cm,scale=1.4]{III} \end{tabular}  \arrow[rightarrow]{ru}{\displaystyle{f=0}} \arrow[rightarrow]{rd} {\displaystyle{T}}&   &    &  \\
 &\begin{tabular}{c} \includegraphics[scale=.8]{Tstar123-1-2} \end{tabular} \arrow[rightarrow]{r} {\displaystyle{f=0}}
 & 
  \begin{tabular}{c} \includegraphics[trim=0cm  0cm 0 0cm,scale=.8]{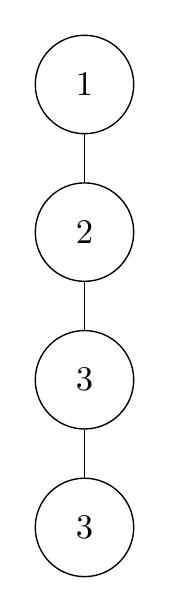} \end{tabular} &   \\
\begin{tabular}{c} \includegraphics[scale=.8]{G2dash}\end{tabular} \arrow[rightarrow]{ru}  {\displaystyle{S}}   \arrow[rightarrow]{r} {\displaystyle{4f^3+27 g^2 s=0}} &\raisebox{-1.5cm}{ \scalebox{.9}{ \includegraphics{TV22-1-1}}} \arrow[rightarrow]{r} {\displaystyle{g=0}} \arrow[rightarrow]{ru}  {\displaystyle{S}} &     \raisebox{-1.5cm}{\includegraphics[scale=.9]{TV123}} & 
\end{tikzcd} 
\caption{ III + I$^{*S_3}_0$,  Resolution IV. \label{fig:IIII0nsRes4}}
\end{table}

\newpage 
\begin{table}
\begin{tikzcd}[column sep=normal]
& \raisebox{-1cm}{\includegraphics[scale=.6]{IVGraph}} \arrow[rightarrow]{r}{\displaystyle{g=0}} \arrow[rightarrow]{dd} {\displaystyle{T}} & \begin{tabular}{c} \includegraphics[scale=.8]{H12} \end{tabular}  &  \\
\begin{tabular}{c} \includegraphics[trim=0cm  1cm 0 0cm,scale=1.4]{III} \end{tabular}  \arrow[rightarrow]{ru}{\displaystyle{{f=0}}} \arrow[rightarrow]{rd} {\displaystyle{T}}&   &    &  \\
 &  \begin{tabular}{c} \includegraphics[trim=0cm  0cm 0 0cm,scale=.6]{TV1242}    \end{tabular}  \begin{tabular}{c}or  \includegraphics[trim=0cm  0cm 0 0cm,scale=.5]{Tstar1-3-421} or \end{tabular}  
 &  \begin{tabular}{c} \includegraphics[trim=2.5cm  0cm 0 0cm,scale=.5]{Tstar12-2-3} or \end{tabular} \   \begin{tabular}{c} \includegraphics[trim=0cm  0cm 0 0cm,scale=.6]{TV1233} \end{tabular}
  \\
\begin{tabular}{c} \includegraphics[scale=.8]{G2dash}\end{tabular} \arrow[rightarrow]{ru}  {\displaystyle{S}}  
 \arrow[rightarrow]{r} {\displaystyle{g=0}} 
 &     \raisebox{-1.5cm}{\includegraphics[scale=.9]{TV123}} & 
\end{tikzcd} 
\caption{ III + I$^{*\mathbb{Z}_3}_0$,  $y^2=x^3+ f s t^{3+\alpha} x + g s^2 t^3$.  Resolution I, II, III, and IV. The fibers in codimension-two are arranged in the order of the resolution. There are no more enhancements in higher codimension. 
\label{Fig:I0Z3}}
\end{table}

\clearpage 
\newpage

\section{Deriving the matter representation of an \sug-model}\label{sec:Rep}
In this section, we explain how to derive the  representation attached to an \sug-model. 
We first explain how  the Katz-Vafa method does not give the correct representation for the \sug-model. 
We use instead the notion of saturation of weights starting with a set of weights derived geometrically by intersecting the fibral divisors by the irreducible components of singular fibers over codimension-two points.

\subsection{Failing of the Katz--Vafa method for the \sug-model}\label{sec:KatzVafa}

The Katz--Vafa method was defined for rank-one enhancements between Lie algebras of type ADE \cite{KatzVafa}.
For a $G$-model with a simple Lie group $G$, the Katz--Vafa method provides an elegant derivation of the matter content. This method had enabled the F-theory constructions to be independent from the existence of a Heterotic dual to explain their matter content.
In collision of singularities, given two simple groups $G_1$ and $G_2$, we can determine the group $G$ that corresponds to their collision by considering an elliptic surface whose discriminant locus pass by one of the points of intersection of the divisors supporting $G_1$ and $G_2$.

The Katz--Vafa method uses as input a triple $(G_1,G_2, G)$ of Lie groups  such that  $G_1\otimes G_2\subset G$  and  returns representations $\mathbf{R}_i$ appearing in the  branching rule for the decomposition of the  adjoint of $\mathfrak{g}=\text{Lie}(G)$ in terms of representations of $G_1\times G_2$.  
$$
\mathfrak{g}=\mathfrak{g}_1\oplus\mathfrak{g}_2\oplus R_1\oplus\cdots\oplus R_n,
$$
where $\mathbf{R}_i$ are representations of $G_1\times G_2$.

The branching rules for $G\supset G_1\otimes G_2$ provide the Katz--Vafa representations $\mathbf{R}_i= (r^{(1)}_i,r^{(2)}_i)$ for the collision of $G_1$ and $G_2$ :
\begin{equation}
adj(\mathfrak{g})=adj(\mathfrak{g}_1)\oplus adj(\mathfrak{g}_2) \oplus \bigoplus_i  (r^{(1)}_i,r^{(2)}_i).
\end{equation}

The original motivation for Katz--Vafa was to give an intrinsic derivation of the matter content of F-theory without relying on the duality with the Heterotic string theory. 
In six-dimensional theories, anomaly cancellations are a sanity check for the identification of the correct representations. But a direct derivation is still necessary as anomalies do not uniquely identify the matter content. In explicit examples, we often have a rank-one enhancement with $G_1\times U(1)\subset G$ with $G_1$ given via the fiber structure and $G$ to be determined. 

The Katz-Vafa method does not always produce the correct matter content, in particular, the matter content of the \sug-model cannot be derived by the Katz-Vafa method. 
 The geometry of the fibration shows that at the collision of $S$ and $T$, the fibers are contractions of fibers of type III$^*$ whose dual graph is the affine Dynkin diagram of type $\widetilde{\text{E}}_7$.
 The Lie algebra of \sug\ is a maximal subalgebra of type $S$ of E$_7$.  The branching rule for  the adjoint representation of E$_7$  in terms of the representation of \sug\ is  as follows \cite{McKayPatera}
: 
 \begin{equation}
 \mathbf{133}= (\mathbf{3,1})\oplus (\mathbf{1,14})\oplus (\mathbf{5,7)}\oplus(\mathbf{3,27}). 
  \end{equation}
The  two summands $(\mathbf{3,1})$ and $(\mathbf{1,14})$  are the two components of the adjoint representation of \sug . 
Thus, the remaining parts $(\mathbf{5,7})\oplus (\mathbf{3,27})$ should give the matter content of the \sug-model according to the Katz--Vafa method.   
 However, the representations $(\mathbf{5,7})$ and  $(\mathbf{3,27})$ are incompatible with the weights from the geometry after resolving the singularities of the \sug-model, as they do not contain the weights of the bifundamental representation $(\mathbf{2,7})$. 

\subsection{Saturation of weights and  representations}\label{sec:Sat}

To determine the representation $\mathbf{R}$ attached to a given elliptic fibration, over each codimension-two point, we  identify a subset of weights of $\mathfrak{g}$ using intersection theory.
This process can be explained in the following steps:
\begin{enumerate}  
\item We first identify all the fibral divisors\footnote{ Under mild assumptions, the discriminant locus of an elliptic fibration $\varphi: Y\to B$ is a divisor $\Delta$. 
The {\em fibral divisors} are by definition the irreducible components of the pullback $\varphi^*({\Delta}_{red})$ where ${\Delta}_{red}$ is the reduced discriminant. 
} of the elliptic fibration. Each fibral divisor is a fibration over an irreducible component of the discriminant locus. 
\item Then, we collect all the  singular fibers over codimension-two points of the base where the generic fibers of the fibral divisors degenerate further. 
Assuming that the fibration is equidimensional up to codimension-two points, all the irreducible components of the singular fibers are necessarily rational curves that we call {\em rational vertical curves}. 
\item We compute the  intersection numbers of these rational vertical curves with the fibral divisors. 
In this way, we attach to each irreducible components of a singular fiber over a codimension-two point a vector with integer entries. 
If we consider only the fibral divisors not touching the section of the elliptic fibration, we get a vector of dimension equal to the rank of the Lie algebra $\mathfrak{g}$; we identify the negative of such a vector with an element of the weight lattice of $\mathfrak{g}$.
\item We determine the representation $\mathbf{R}$ attached to the weights of the rational vertical curves over each codimension-two points. 
The representation $\mathbf{R}$ has to necessarily contain all the weights that we have identified by intersections of rational vertical curves with the fibral divisors.
\end{enumerate}

In the last step, the goal is then to determine the minimal representation generated by the subset of weights given to us by the geometry.
  A systematic way to do so is  to use the notion  of  {\em saturation of weights} introduced by Bourbaki (see \cite[Chap.VIII.\S 7. Sect. 2]{Bourbaki.GLA79}).  
  The identification of a representation using weights obtained by intersection numbers can be traced back to  Aspinwall and Gross \cite{Aspinwall:1996nk}; it is  based on the M-theory picture of M2-branes wrapping chains of curves and becoming massless when the curve shrinks to a point. The use of saturation of weights in this context was first explained in \cite{G2} and inspired directly by \cite{Aspinwall:1996nk} and \cite{Bourbaki.GLA79}. See also \cite{Marsano:2011hv,Morrison:2011mb}.

  \begin{defn}
 A set $\Pi$ of integral weights is said to be \emph{saturated} if for all  $\mu\in \Pi$, for all roots $\alpha$, and for all $t$ between $0$ and  $\langle \mu, \alpha \rangle $, we have   $\mu-t \alpha\in \Pi$. The {\em saturation} of a set of weights $\Pi$ is the smallest saturated set containing $\Pi$. 
 \end{defn}
   We refer to Bourbaki \cite[Chap.VIII.\S7. Sect.2]{Bourbaki.GLA79} for more information on saturated sets of weights. 
   A direct consequence of the definition is that a saturated set of weights is invariant under the Weyl group. 
Given a representation $\mathbf{R}$, its set of weights  $\Phi_{\mathbf{R}}$  is a saturated set of weights.     

 \begin{thm}[{See Corollary of Proposition 5 of \cite[Chap.VIII.\S7. Sect.2]{Bourbaki.GLA79}}]
 For any finite saturated set of weights $\Pi$, there exists a finite-dimensional representation $\bf{R}$ such that $\Pi$ is the set of weights of $\mathbf{R}$.
 \end{thm}
   Thus, the saturation of weights  attaches a representation $\mathbf{R}$ to a finite set of weights $\Pi$ in a systematic way: $\mathbf{R}$ is the unique representation whose set of weights is the saturation of  $\Pi$.   
In other words, $\mathbf{R}$ is the representation generated by the set of weights $\Pi$. There can be other representations containing the saturation of $\Pi$ as a proper subset of their set of weights.

We determine the representation $\mathbf{R}$ attached to an elliptic fibration by first interpreting the negative of the intersection of vertical curves (over codimension-two points) with the fibral divisors as elements of the weight lattice.  
We then associate a representation to these weights using the notion of saturation of weights \cite{F4,G2,EKY2}. 

We illustrate the method with Resolution I.  
In Resolution I,  the curve $C_1^\text{s}$ degenerates as follows over $S\cap T$:
$$
 \begin{tikzcd}  C_1^{\text{s}} \arrow[rightarrow]{r}  & 3\eta_1^{02}+2\eta_1^{0A}+2\eta_1^{12}+\eta_1^2,  \end{tikzcd}
 $$
 where $\eta_1^{02}$ and $\eta_1^2$  have weight $[0;2,-1]$ , which is a weight of the  fundamental representation $(\bf{1,7})\subset (\bf{1,14})$, $\eta_1^{0A}$ has weight $[1;-1,0]$ of the bifundamental representation $(\bf{2}, \bf{7})\subset (\bf{2,14})$, and $\eta_1^{12}$ has  weight $[0;-3,2]$  of the adjoint representation $(\bf{1,14})$.
 
 In Resolution IV, the curve $C_1^\text{s}$ undergoes the following degeneration over $S\cap T$: 
  $$
 \begin{tikzcd}  C_1^{\text{s}} \arrow[rightarrow]{r}  & \eta_{01}^2+\eta_1^2,  \end{tikzcd}\
 $$
 where both $\eta_{01}^2$ and $\eta_1^2$ have weight  $[1;0,0]$, which is a weight of the representation $(\bf{2,1})\subset(\bf{2,7})$. 
There is another curve ($\eta_{0}^{2B}$) coming from the degeneration of $C_0^\text{s}$ that carries the weight $[-1,2,-1]$ of the bifundamental $(\bf{2,7})$ of \sug . 

 To make sense of the representation we should attach to these degenerations, we recall first few fact about the weights involved. 
The adjoint representation of G$_2$ consists of fourteen weights that split into four orbits under the action of the Weyl group: each of the two zero weights form an orbit on its own,  the six short roots form an orbit, and the six long weights form the fourth orbit. 
The short roots of the adjoint of G$_2$ are the nonzero weights of the fundamental representation $\bf{7}$ of G$_2$.  
The highest weight of the adjoint representation is a long root; thus,  the adjoint representation is the smallest representation containing the saturation set of any subset of long roots. 
The bifundamental representation of \sug\  $(\bf{2,7})$ consists of the following  Weyl orbits: 
the two weights $[\pm 1;0,0]$ form an orbit corresponding to the representation $(\bf{2,1})$, and the other twelve weights form a unique orbit.  
Thus we can conclude the saturations from the weights are given by Table \ref{table:sat}.

\begin{table}[h!]
\begin{center}
\begin{tabular}{|c|c|}
\hline 
Set of weights $\Pi$ & Saturation of $\Pi$ \\ 
\hline 
$\{[0;2,-1]\}$  &  the fundamental rep.  $(\bf{1,7})$\\
 $\{[1;-1,0]\}$ &the bifundamental rep.  $(\bf{2}, \bf{7})$\\
$\{[0;-3,2]\}$ & the adjoint representation $(\bf{1,14})$\\
 $\{[1;0,0]\}$ &   the fundamental  rep.  $(\bf{2,1})$\\
 $\{ [0;-3,2], [0;2,-1]\}$ & the adjoint representation  $(\bf{1,14})$\\
 $\{[-1;0,0], [-1;2,-1]\}$ &  the bifundamental rep.  $(\bf{2,7})$\\
 \hline
\end{tabular}
\end{center}
\caption{ Representations derived  from saturation of weights. }
\label{table:sat}
\end{table}

\section{$5d$ and $6d$ Supergravity Theories with Eight Supercharges}\label{sec:Physics}
 In this section, we determine the prepotential in each Coulomb branch of the five-dimensional supergravity theory with eight supercharges and match with the triple intersection polynomial for each corresponding crepant resolution. We determine the number of multiplets of the corresponding uplifted six-dimensional theory and check if the anomalies can be canceled using Green-Schwarz mechanism. We also explain the numerical oddities we encounter while counting the number of hypermultiplets.
 
\subsection{ $5d$ $\mathcal{N}=1$ prepotentials}

In the Coulomb phase of an $\mathcal{N}=1$ supergravity theory in five dimension, 
 the scalar fields of the vector multiplets are restricted to the Cartan sub-algebra of the Lie group as the Lie group is broken to $U(1)^r$ where $r$ is the rank of the group. It follows that the charge of an  hypermultiplet is simply given by a  
 weight of the representation under which it transforms \cite{IMS}.  
The Intrilligator-Morrison-Seiberg (IMS) prepotential is the quantum contribution to the prepotential of a five-dimensional gauge theory with the matter fields in the representations $\mathbf{R}_i$ of the gauge group. Let $\phi$ be in the Cartan subalgebra of a Lie algebra $\mathfrak{g}$. 
The  weights are in the dual space of the Cartan subalgebra.  We denote the evaluation of a  weight on a coroot vector $\phi$ as a scalar product $\langle \mu,\phi \rangle$.  We recall that the roots are the weights of the adjoint representation of  $\mathfrak{g}$.
Denoting the fundamental simple roots by $\alpha$ and the weights of $\mathbf{R}_i$ by $\varpi$ we have \cite{IMS}
\begin{align}
6\mathscr{F}_{\text{IMS}}(\phi) =&\frac{1}{2} \left(
\sum_{\alpha} |\langle \alpha, \phi \rangle|^3-\sum_{i} \sum_{\varpi\in \mathbf{R}_i} n_{\mathbf{R}_i} |\langle \varpi, \phi\rangle|^3 
\right).\label{Eq:IMS}
\end{align}
For all simple groups with the exception of SU$(N)$ with $N\geq 3$, this is the full cubic prepotential as there are non-trivial third Casimir invariants. 

The open dual fundamental Weyl chamber is defined as the cone $ \langle \alpha, \phi \rangle>0$, where $\alpha$ runs through the set of all simple positive roots. 
 For a given choice of a group $G$ and representations $\mathbf{R}_i$, we have to determine a Weyl chamber to remove the absolute values in the sum over the roots. 
 We then consider the  hyperplane arrangement  $\langle \varpi, \phi\rangle=0$, where $\varphi$ runs through all the weights of all the representations $\mathbf{R}_i$ and $\phi$ is an element of the coroot space. 
 If none of these hyperplanes intersect the interior of the dual Weyl chamber of $\mathfrak{g}$, we can safely remove the absolute values in the sum over the weights. 
 Otherwise, we have hyperplanes partitioning the dual fundamental Weyl chamber into subchambers. Each of these subchambers is defined by the signs of the linear forms $\langle \varpi, \phi\rangle$. 
 Two such subchambers are adjacent when they differ by the sign of a unique linear form.  
 Within each of these subchambers, the prepotential is a cubic polynomial. 
 But as we go from one subchamber to an adjacent one, we have to go through one of the walls defined by the weights. 
 The transition from one chamber to an adjacent chamber is a phase transition.

We compute $\mathscr{F}_{\text{IMS}}$ for each of the four chambers of an \sug-model. The chambers are defined by Table \ref{Table:Phases}. 
\begin{thm}\label{Thm:Prepotentials}
The prepotential of an \sug-model in the four phases defined by the 
chambers of Table \ref{Table:Phases}.
\quad

\begin{description}
\item[Chamber 1:]
\begin{align}\nonumber
\begin{aligned}
6\mathscr{F}^{(1)}_{\text{IMS}}=& -4 \phi_2^3 (2 n_{\bf{1,14}}+n_{\bf{1,7}}-2)+9 \phi_1 \phi_2^2 (-2 n_{\bf{1,14}}+n_{\bf{1,7}}+2) -8 (n_{\bf{1,14}}-1) \phi_1^3\\
& +3 \phi_1^2 \phi_2 (8 n_{\bf{1,14}}-n_{\bf{1,7}}-8) + \psi _1^3 (-n_{\bf{2,1}}-7 n_{\bf{2,7}}-8 n_{\bf{3,1}}+8) \\
& +12 \psi _1 \left(-3 n_{\bf{2,7}} \phi_2^2+3 n_{\bf{2,7}} \phi_1 \phi_2- n_{\bf{2,7}} \phi_1^2\right),
\end{aligned}
\end{align}
\item[Chamber 2:]
\begin{align}\nonumber
\begin{aligned}
6\mathscr{F}^{(2)}_{\text{IMS}}=& -2 \phi_2^3 (4 (n_{\bf{1,14}}+n_{\bf{1,7}}-1)+n_{\bf{2,7}})+9 \phi_1 \phi_2^2 (-2 n_{\bf{1,14}}+n_{\bf{1,7}}+2) -8 (n_{\bf{1,14}}-1) \phi_1^3 \\
& +3 \phi_1^2 \phi_2 (8 n_{\bf{1,14}}-n_{\bf{1,7}}-8) +\psi _1^3 (-n_{\bf{2,1}}-5 n_{\bf{2,7}}-8 n_{\bf{3,1}}+8)-6 n_{\bf{2,7}} \psi _1^2 \phi_2 \\
& +\psi _1 \left(-30 n_{\bf{2,7}} \phi_2^2+36 n_{\bf{2,7}} \phi_1 \phi_2-12 n_{\bf{2,7}} \phi_1^2\right),
\end{aligned}
\end{align}
\item[Chamber 3:]
\begin{align}\nonumber
\begin{aligned}
6\mathscr{F}^{(3)}_{\text{IMS}}=& \ 3 \phi_1 \phi_2^2 (-6 n_{\bf{1,14}}+3 n_{\bf{1,7}}-2 n_{\bf{2,7}}+6)+3 \phi_1^2 \phi_2 (8 n_{\bf{1,14}}-n_{\bf{1,7}}+2 n_{\bf{2,7}}-8) \\
& -8 \phi_2^3 (n_{\bf{1,14}}+n_{\bf{1,7}}-1)-2 \phi_1^3 (4 n_{\bf{1,14}}+n_{\bf{2,7}}-4)+\psi _1^3 (-n_{\bf{2,1}}-3 n_{\bf{2,7}}-8 n_{\bf{3,1}}+8) \\
& + \psi _1 \left(-24 n_{\bf{2,7}} \phi_2^2+24 n_{\bf{2,7}} \phi_1 \phi_2-6 n_{\bf{2,7}} \phi_1^2\right)-6 n_{\bf{2,7}} \psi _1^2 \phi_1,
\end{aligned}
\end{align}
\item[Chamber 4:]
\begin{align}\nonumber
\begin{aligned}
6\mathscr{F}^{(4)}_{\text{IMS}}=&-8 (n_{\bf{1,14}}-1) \phi_1^3 + 9 \phi_1 \phi_2^2 (-2 n_{\bf{1,14}}+n_{\bf{1,7}}+2 n_{\bf{2,7}}+2)-8 \phi_2^3 (n_{\bf{1,14}}+n_{\bf{1,7}}+2 n_{\bf{2,7}}-1) \\
& +3 \phi_2 \phi_1^2 (8 n_{\bf{1,14}}-n_{\bf{1,7}}-2 n_{\bf{2,7}}-8) +\psi _1^3 (-n_{\bf{2,1}}-n_{\bf{2,7}}-8 n_{\bf{3,1}}+8)-12 n_{\bf{2,7}} \psi _1^2 \phi_2.
\end{aligned}
\end{align}
\end{description}
\end{thm}
\begin{proof}
Direct computation starting with equation \eqref{Eq:IMS} and using Table   \ref{Table:Phases} to remove the absolute values.
\end{proof}

The number of hypermultiplets are computed by comparing the prepotential and the intersection polynomial. Comparing the triple intersection numbers obtained in the Resolutions  I, II, III, IV with the prepotentials computed respectively in Chambers 1,2,3,4, we get
\begin{align}\label{eq:numbers}
\begin{aligned}
&  n_{\mathbf{2,1}}+ 8 n_{\mathbf{3,1}}=S\cdot \left(4 L+2 S-\frac{7}{2} T\right)+8 , \quad &&n_{\mathbf{1,14}} = \frac{1}{2} (-L \cdot T + T^2 + 2), \\
& n_{\mathbf{2,7}} = \frac{1}{2} S \cdot T,  \quad  &&n_{\mathbf{1,7}}=T \cdot (5 L - S - 2 T). 
\end{aligned}
\end{align}
We see in particular that the numbers $ n_{\mathbf{2,1}}$ and $n_{\mathbf{3,1}}$ are  not completely fixed by this method but restricted by a linear relation.   
The same is true for the SU($2$)-model in \cite{ES}. 
 Using Witten's genus formula  to restrict $n_{\mathbf{3,1}}$ with $K=-L$, we get that 
$ n_{\mathbf{3,1}}$ and $ n_{\mathbf{1,14}}$ become respectively the arithmetic genus of the curves $S$ and $T$:
\begin{equation}\label{Eq.5d.hyp}
n_{\mathbf{3,1}} =g(S),\quad n_{\mathbf{1,14}}= g(T), \quad n_{\mathbf{2,1}} =-S\cdot \left(8K+2 S+\frac{7}{2}T\right),\quad n_{\mathbf{1,7}}= -T\cdot (5 K+ S +2 T).
\end{equation}

The geometric origin of  the number $n_{\bf{1,7}}$ is as follows. 
The G$_2$ fiber contains a non-split curve that splits into three curves after an appropriate field extension. Over $T$, this non-split curve defines a triple cover of $T$ with a branch locus given by its discriminant $  s^3 t^6 (4 f^3+27 g^2 s)$.  The reduced discriminant is $s(4 f^3+27 g^2 s)$  as $t$ is a unit. Using Witten's genus formula \cite{Witten:1996qb,Aspinwall:2000kf,Witten:1996qb}, we get 
\begin{equation}
n_{\bf{1,7}}=(d-1) (g-1) +\frac{1}{2} R,
\end{equation}
where $d=3$, $g=\frac{1}{2}(K T + T^2+2)$, and the class of $R$ is the class of the reduced discriminant, i.e. $R=[t][s(4 f^3+27 g^2 s)]=-3T\cdot (4 K+S+2T)+S\cdot T$. 

As explained in Remark \ref{rem1}, the representation $n_{\bf{1,7}}$ is frozen when $S$ and $T$ intersect transversally at one point and $T$ is a $-3$-curve. 
The meaning of  $n_{\bf{2,1}}$ is much more complicated. One would expect it to be just an intersection number, but that is not the case as will discuss in section \ref{sec:odd}.

\subsection{Generalities on $6d$ $\mathcal{N}=(1,0)$  anomaly cancellations}

We consider a gauged six-dimensional $\mathcal{N}=(1,0)$ supergravity theory \cite{Green:1984bx} with a semi-simple gauge group $G=\prod_a G_a$, $n_V^{(6)}$ vector multiplets, $n_T$ tensor multiplets, and $n_H$ hypermultiplets consisting of $n_H^0$ neutral hypermultiplets and $n_H^{ch}$ charged hypermultiplets under a representation $\bigoplus_i\mathbf{R}_i$ of the gauge group with $\mathbf{R}_i=\bigotimes_{a} \mathbf{R}_{i,a}$, where $\mathbf{R}_{i,a}$ is an irreducible representation of the simple component  $G_a$ of the semi-simple group $G$. The vector multiplets transform under the adjoint of the gauge group.  CPT invariance requires the representation to be quaternionic, and hence we have quaternions of dimension $n_H$ for the hypermultiplets.

In particular, F-theory compactified  on a Calabi-Yau threefold $Y$ gives a six-dimensional supergravity theory with eight supercharges coupled to n$_V$ vector, $n_T$ tensor,  and $n_H^0=h^{2,1}(Y)+1$ neutral hypermultiplets \cite{Cadavid:1995bk}. 
When the Calabi-Yau variety is elliptically fibered  with a gauge group $G$ and a representation $\mathbf{R}$, the number of vectors is $n_V=\dim G$, the number of tensors is $n_T=9 -K^2$,  and we have charged hypermultiplets transforming in the representation $\mathbf{R}$ of $G$. 
 The base of the fibration is then necessarily a rational surface $B$ whose canonical class is denoted $K$ \cite{Sadov:1996zm}.
 For anomalies in six dimensions, we used  \cite{Green:1984bx,Sadov:1996zm,GM1,Kumar:2010ru,Salam, Monnier:2017oqd} and we refer to \cite{Park} for an elegant  general review.

Since we consider only local anomalies, we only have pure gravitational anomalies, pure gauge anomalies, and mixed gravitational and gauge anomalies. In order to address these anomalies, we use Green-Schwarz mechanism in six-dimensions. The anomaly polynomial I$_8$ has a pure gravitational contribution of the form  $\propto (n_H-n_V^{(6)}+29n_T-273)\tr{}  R^4$ where $R$ is the Riemann tensor thought of as a $6\times 6$ real matrix of two-form values. 
To apply the Green-Schwarz mechanism, the coefficient of $\tr{}  R^4$ is required to vanish \cite{Salam}:
\begin{equation}
n_H-n_V^{(6)}+29n_T-273=0.
\end{equation}
The remainder terms of the anomaly polynomial I$_8$, is given by
\begin{equation}
I_8=\frac{K^2}{8} (\tr R^2)^2 +\frac{1}{6}\sum_{a} X^{(2)}_{a} \tr R^2-\frac{2}{3}\sum_{a} X^{(4)}_{a}+4\sum_{a<b}Y_{ab},
\end{equation}
where each gauge group contributions $X^{(n)}_{a}$ for $n=2,4$ and the mixed contribution $Y_{ab}$ are given by 
\begin{equation}
X^{(n)}_{a}=\tr_{\bf{adj}}F^n_a -\sum_{i}n_{\bf{R}_{i,a}}\tr_{\bf{R}_{i,a}}F^n_a, \quad 
Y_{ab}=\sum_{i,j} n_{\bf{R}_{i,a}, \bf{R}_{j,b}} \tr_{\bf{R}_{i,a}}F^2_a \tr_{\bf{R}_{j,b}}F^2_b,
\end{equation}
and $n_{\bf{R}_{i,a}, \bf{R}_{j,b}}$ is the number of hypermultiplets transforming in the representation $(\mathbf{R}_{i,a},\mathbf{R}_{j,b})$ of $G_a\times G_b$.

It is important to note that when a representation is charged on more than a simple component of the group, it affects not only  $Y_{ab}$ but also $X^{(2)}_a$ and $X^{(4)}_a$. Consider a representation $(\bf{R_1},\bf{R_2})$ for of a semisimple group with two simple components $G=G_1\times G_2$, where $\bf{R_a}$ is a representation of $G_a$. Then this representation contributes $\dim{\bf{R_2}}$ times   to $n_{\bf{R_1}}$, and contributes  $\dim{\bf{R_1}}$ times to $n_{\bf{R_2}}$:
$$
n_{\bf{R_1}}=\cdots+\dim{\bf{R_2}} \ n_{\bf{R_1},\bf{R_2}}, \quad n_{\bf{R_2}}=\cdots +\dim{\bf{R_1}} \ n_{\bf{R_1},\bf{R_2}}.
$$

By denoting the zero weights of a representation $\bf{R}_i$ as $\bf{R}^{(0)}_i$, the charged dimension of the hypermultiplets in representation $\bf{R}_i$ is given by $\dim{\bf{R}_i}-\dim{\bf{R}_{i}^{(0)}}$ since the hypermultiplets of zero weights are considered neutral. For a representation $\bf{R}_i$, $n_{\bf{R}_i}$ denotes the multiplicity of the representation $\bf{R}_i$. Then the number of charged hypermultiplets is as follows \cite{GM1}
\begin{equation}
n_H^{ch}=\sum_{i} n_{\bf{R}_i} \left( \dim{\bf{R}_i}-\dim{\bf{R_{i}^{(0)}}} \right).
\end{equation}

The trace identities for a representation $\mathbf{R}_{i,a}$ of a simple group $G_a$ are
\begin{equation}
\tr_{\bf{R}_{i,a}} F^2_a=A_{\bf{R}_{i,a}} \tr_{\bf{F}_a} F^2_a , \quad \tr_{\bf{R}_{i,a}} F^4_a=B_{\bf{R}_{i,a}} \tr_{\bf{F}_a} F^4_a+C_{\bf{R}_{i,a}} (\tr_{\bf{F}_a} F^2_a )^2
\end{equation}
with respect to a  reference representation $\bf{F}_a$ for each simple component $G_a$.\footnote{We denoted this representation as $\bf{F}_a$ as we picked the fundamental representations. However, any representation can be used as a reference representation.}
These  group theoretical coefficients are listed in \cite{Erler,Avramis:2005hc,vanRitbergen:1998pn}.

We then have 
\begin{align}
X^{(2)}_a&=\left(A_{a,\bf{adj}}-\sum_{i}n_{\bf{R}_{i,a}}A_{\bf{R}_{i,a}}\right)\tr_{\bf{F}_a}F^2_a, \\
X^{(4)}_a&=\left(B_{a,\bf{adj}}-\sum_{i}n_{\bf{R}_{i,a}}B_{\bf{R}_{i,a}}\right)\tr_{\bf{F}_a}F^4_a
+\left(C_{a,\bf{adj}}-\sum_{i}n_{\bf{R}_{i,a}}C_{\bf{R}_{i,a}}\right)(\tr_{\bf{F}_a}F^2_a)^2 , \\
Y_{ab}&=\sum_{i,j} n_{\bf{R}_{i,a},\bf{R}_{j,b}} A_{\bf{R}_{i,a}} A_{\bf{R}_{j,b}} \tr_{\bf{F}_a}F^2_a \tr_{\bf{F}_b}F^2_b, \quad (a\neq b).
\end{align}
For each simple component $G_a$, the anomaly polynomial  I$_8$ has a pure gauge contribution proportional to the quartic term $\tr F^4_a$ that is required to vanish in order to factorize I$_8$: 
$$
B_{a,\bf{adj}}-\sum_{i}n_{\bf{R}_{i,a}}B_{\bf{R}_{i,a}}=0.
$$
When the coefficients of all quartic terms  ($\tr R^4$ and $\tr F^4_a$) vanish,  the remaining part of the anomaly polynomial I$_8$ is
\begin{align}
\begin{cases}
&\mathrm{I}_8=\frac{K^2}{8} (\tr R^2)^2 +\frac{1}{6}\sum_{a} X^{(2)}_{a} \tr R^2-\frac{2}{3}\sum_{a} X^{(4)}_{a}+4\sum_{a<b}Y_{ab}, \\
&X^{(2)}_a=\left(A_{a,\bf{adj}}-\sum_{i}n_{\bf{R}_{i,a}}A_{\bf{R}_{i,a}}\right)\tr_{\bf{F}_a}F^2_a, \quad
X^{(4)}_a=\left(C_{a,\bf{adj}}-\sum_{i}n_{\bf{R}_{i,a}}C_{\bf{R}_{i,a}}\right)(\tr_{\bf{F}_a}F^2_a)^2, \\
&Y_{ab}=\sum_{i,j} n_{\bf{R}_{i,a},\bf{R}_{j,b}} A_{\bf{R}_{i,a}} A_{\bf{R}_{j,b}} \tr_{\bf{F}_a}F^2_a \tr_{\bf{F}_b}F^2_b, \quad (a\neq b).
\end{cases}
\end{align}
The anomalies are canceled by the Green-Schwarz mechanism when I$_8$ factorizes \cite{Green:1984bx,Sagnotti:1992qw,Schwarz:1995zw}. 
The modification of the field strength $H$ of  the antisymmetric tensor $B$ is 
\begin{equation}
H= dB + \frac{1}{2} K \omega_{3L}+ 2\sum_a\frac{S_a}{\lambda_a}\omega_{a,3Y}, 
\end{equation}
where  $\omega_{3L}$ and $\omega_{a,3Y}$ are respectively the gravitational and Yang-Mills Chern-Simons  terms. 
 If I$_8$ factors as 
 \begin{equation}
 \text{I}_8= X\cdot  X,
 \end{equation}
  then the anomaly is canceled by adding the following Green-Schwarz counter-term 
\begin{equation}
\Delta L_{GS}\propto \frac{1}{2} B \wedge X,
\end{equation}
which implies that $X$ carries string charges.  Following Sadov \cite{Sadov:1996zm}, to cancel the local anomalies, we consider
\begin{equation}
X= \frac{1}{2} K \tr R^2 +\sum_a \frac{2}{\lambda_a} S_a\tr F^2_a,
\label{eq:Xfactor}
\end{equation} where the  $\lambda_a$ are normalization factors  chosen such that the  smallest
topological charge of an embedded SU($2$) instanton in G$_a$ is one \cite{Kumar:2010ru, Park, Bernard}. 
This forces $\lambda_a$ to be the Dynkin index of the fundamental representation of  $G_a$ as summarized in Table \ref{tb:normalization} \cite{Park}. 
\begin{table}[h!]
\begin{center}
\begin{tabular}{|c|c|c|c|c|c|c|c|c|c|}
\hline
 $\mathfrak{g}$ & A$_n$ & B$_n$ & C$_n$ & D$_n$ & E$_8$ & E$_7$ & E$_6$&  F$_4$ & G$_2$ \\
 \hline
 $\lambda$ & $1$ & $2$  & $1$ & $2$ & $60$ & $12$ & $6$ & $6$ & $2$ \\
 \hline  
\end{tabular}
\caption{The normalization factors for each simple gauge algebra. See \cite{Kumar:2010ru}.}
\label{tb:normalization}
\end{center}
\end{table}

If the simple group $G_a$ is supported on a divisor $S_a$, the local anomaly cancellation conditions are  the following equations
\begin{subequations}
\begin{align}
n_T&=9-K^2 , \\
n_H-n_V^{(6)}+29n_T-273 &=0,\\
\left(B_{a,\bf{adj}}-\sum_{i}n_{\bf{R}_{i,a}}B_{\bf{R}_{i,a}}\right)& = 0, \\
\lambda_a  \left(A_{a,\bf{adj}}-\sum_{i}n_{\bf{R}_{i,a}}A_{\bf{R}_{i,a}}\right) & =6  K\cdot S_a, \\
\lambda_a^2 \left(C_{a,\bf{adj}}-\sum_{i}n_{\bf{R}_{i,a}}C_{\bf{R}_{i,a}}\right) & =-3 S_a ^2, \\
\lambda_a \lambda_b \sum_{i,j} n_{\bf{R}_{i,a},\bf{R}_{j,b}} A_{\bf{R}_{i,a}} A_{\bf{R}_{j,b}}  & =S_a\cdot S_b, \quad (a\neq b).
\end{align}
\end{subequations}

Assuming the first three equations hold, 
cancelling the anomalies is equivalent to factoring the  anomaly polynomial 
\begin{equation}
I_8 =\frac{K^2}{8} (\tr R^2)^2 +\frac{1}{6} (X^{(2)}_{1} +X^{(2)}_{2}) \tr R^2-\frac{2}{3} (X^{(4)}_{1}+X^{(4)}_{2})+4Y_{12} .
\end{equation}

The total number of hypermultiplets is the sum of the neutral hypermultiplets and the charged hypermultiplets. 
For a compactification on a Calabi-Yau threefold $Y$, the  number of  neutral hypermultiplets is  $h^{2,1}(Y)+1$ \cite{Cadavid:1995bk}. The number of each multiplet is \cite{GM1}
\begin{align}
n_V^{(6)}&=\dim{G}, \quad n_T=h^{1,1}(B)-1=9-K^2 , \\
n_H&=n_H^0+n_H^{ch}=h^{2,1}(Y)+1+\sum_{i} n_{\bf{R}_i} \left( \dim{\bf{R}_i}-\dim{\bf{R_{i}^{(0)}}} \right),
\end{align}
where the (elliptically-fibered) base $B$ is a rational surface.

\subsection{Anomalies cancellations for an \sug-model}

In this section, we check that the gravitational, gauged, and mixed anomalies of the six-dimensional supergravity are all canceled when the Lie algebra and the representation are
$$\mathfrak{g}= \text{A}_1\oplus \text{G}_2, \quad \bold{R}=(\bold{3},\bold{1})\oplus (\bold{1},\bold{14})\oplus (\bold{2},\bold{1})\oplus(\bold{1},\bold{7})\oplus(\bold{2},\bold{{7}}).$$
First, we recall that for the case of a Calabi-Yau threefold (defined as a crepant resolution of the Weierstrass model of an \sug-model), the Euler characteristic is (see Lemma \ref{lemEulerCY3})
\begin{equation}
\chi(Y)=-6 (10 K^2+5 K\cdot S+8 K\cdot T+S^2+2 S\cdot T+2 T^2),
\end{equation}
 where $S$ and $T$ are the curves supporting  $\text{A}_1$ and $\text{G}_2$ respectively. 
The Hodge numbers are
\begin{equation}
h^{1,1}(Y)=14-K^2, \quad h^{2,1}(Y)=29 K^2+15 K\cdot S+24 K \cdot T+3 S^2+6 S \cdot T+6 T^2+14 .
\end{equation}
The numbers of vector multiplets and tensor multiplets, and neutral hypermultiplets are
\begin{align}\label{Eq:An}
\begin{aligned}
& n_T=9-K^2, \quad n_V=\dim G=\dim\  \text{SU}(2)+\dim\ \text{G}_2=3+14=17, \\
&  n_H^0=h^{2,1}(Y)+1=29 K^2+15 K\cdot S+24 K\cdot  T+3 S^2+6 S \cdot T+6 T^2+15.
\end{aligned}
\end{align}

We will use the anomaly cancellation conditions to explicitly compute the number of hypermultiplets transforming in each representation by requiring all anomalies to cancel.
We will see that they are the same as those found in five-dimensional supergravity  by comparing the triple intersection numbers of the fibral divisors of a given crepant resolution and the cubic prepotential of the corresponding  Coulomb chamber.  

The Lie algebra of type A$_1$ (resp. G$_2$) only has a unique  quartic Casimir invariant so that we do not have to impose the vanishing condition for the coefficients of $ \mathrm{tr}\  F^4_1$  (resp. $ \mathrm{tr}\  F^4_2$). 
We have the following  trace identities for SU($2$) and G$_2$ \cite{Erler,Avramis:2005hc}:
\begin{align}
\begin{aligned}
& \mathrm{tr}_{\bf{3}}F^2_1=4 \mathrm{tr}_{\bf{2}}F^2_2, \quad  \mathrm{tr}_{\bf{3}} F^4_1=8 ( \mathrm{tr}_{\bf{3}} F^2_1)^2, \quad  \mathrm{tr}_{\bf{2}}F^4_1=\frac{1}{2} ( \mathrm{tr}_{\bf{2}}F^2_1)^2,\\
&\mathrm{tr}_{\mathbf 14} F^2_2 = 4 \mathrm{tr}_{\mathbf 7}\ F^2_2 ,  \quad 
 \mathrm{tr}_{\mathbf 14} F^4_2 = \frac{5}{2} (\mathrm{tr}_{\mathbf 7}\ F^2_2)^2,
 \quad
\mathrm{tr}_{\mathbf 7} F^4_2=\frac{1}{4}  (\mathrm{tr}_{\mathbf 7}\ F^2_2)^2,
\end{aligned}
\end{align}
which give 
\begin{align}
\begin{split}
X^{(2)}_1 &= (4 -4 n_{\bf{3,1}} - n_{\bf{2,1}} - 7 n_{\bf{2,7}}) \ { \mathrm{tr}}_{\bf{2}} F^2_1,\quad 
X^{(2)}_2 =( 4 - 4n_{\bf{1,14}} - n_{\bf{1,7}} - 2 n_{\bf{2,7}})\ { \mathrm{tr}}_{\bf{7}} F^2_2,\\
X^{(4)}_1 &= (8  -8 n_{\bf{3,1}} - \frac{1}{2} n_{\bf{2,1}} - \frac{7}{2}n_{\bf{2,7}}) ( \mathrm{tr}_{\bf{2}} F^2_1)^2,\quad
X^{(4)}_2 =   (\frac{5}{2} -\frac{5}{2} n_{\bf{1,14}}-\frac{1}{4} n_{\bf{1,7}}-\frac{2}{4}n_{\bf{2,7}})(\mathrm{tr}_{\bf{7}} F^2_2)^2,\\
Y_{\bf{27}} \ &= n_{\bf{2,7}}  \ { \mathrm{tr}}_{\bf{2}} F^2_1 \ { \mathrm{tr}}_{\bf{7}} F^2_7.
\end{split}
\end{align}
Following Table \ref{tb:normalization}, we  take $\lambda_1=1$ for A$_1$ and $\lambda_2=2$ for G$_2$.
 If the pure gravitational anomaly vanishes, the anomaly cancellation conditions are 
\begin{align}
\begin{aligned}
2( 4 - 4n_{\bf{1,14}} - n_{\bf{1,7}} - 2 n_{\bf{2,7}}) &=  6  K \cdot T,\quad 
  (4 -4 n_{\bf{3,1}} - n_{\bf{2,1}} - 7 n_{\bf{2,7}}) = 6 K \cdot S, \\
 (10 -10 n_{\bf{1,14}}- n_{\bf{1,7}}-2n_{\bf{2,7}})&=-3  T^2,\quad
(8  -8 n_{\bf{3,1}} - \frac{1}{2} n_{\bf{2,1}} - \frac{7}{2}n_{\bf{2,7}}) =-3 S^2,\\
 2n_{\bf{2,7}} &= S \cdot T.
 \end{aligned}
\end{align}
These linear equations  have the following unique solution
\begin{align}\label{eq:numbersReal}
\begin{aligned}
n_{\bf{2,7}} &= \frac{1}{2}S\cdot T, \quad &&  n_{\bf{3,1}} = \frac{1}{2}(K\cdot S+S^2+2) &&  \quad  n_{\bf{2,1}} &&=-S\cdot (8 K +2S+ \frac{7}{2}T), \\
& \quad  
   &&   n_{\bf{1,14}}=\frac{1}{2}(K\cdot T+T^2+2), && \quad n_{\bf{1,7}} &&= -T\cdot (5 K+S+2 T).
 \end{aligned}
 \end{align}
We note that the numbers of charged hypermultiplets match exactly the values found in the five-dimensional theory (see equation \eqref{Eq.5d.hyp}) after using Witten's genus formula which asserts that the number of adjoint hypermultiplets $n_{\bf{3,1}}$ is the genus of the curve $S$ supporting the group SU($2$)  \cite{Witten:1996qb}.
 
 We can now check that the pure gravitational anomaly is cancelled.
The total number of hypermultiplets is the sum of the neutral hypermultiplets coming from the compactification and the charged hypermultiplets transforming under the different irreducible summands of the representation $\bf{R}$. Since the 
charge of a hypermultiplet is given by a weight of a representation, we remove the zero weights when counting charged hypermultiplets \cite{GM1}. In the present case, we have:
\begin{align}
\begin{aligned}
n_H & =n_H^0+n_H^{ch}\\
& =(h^{2,1}(Y)+1)+2 n_{\bf{2,1}}+(7-1) n_{\bf{1,7}}+14 n_{\bf{2,7}}+(3-1) n_{\bf{3,1}}+(14-2) n_{\bf{1,14}}=29 K^2+29.
\end{aligned}
\end{align}

 Using equation \eqref{Eq:An}, we  check that the  coefficient of $\mathrm{tr}\ R^4$ vanishes as required by the cancellation of the pure gravitational anomaly \cite{Salam}:
\begin{equation}
n_H-n_V+29n_T-273=0.
\end{equation}
Finally, we show that the  anomaly polynomial I$_8$ factors as a perfect square:
\begin{align}
\begin{split}
I_8 &= \frac{K^2}{8} ( \mathrm{tr} R^2)^2 + \frac{1}{6} (X_1^{(2)} + X_2^{(2)})  \mathrm{tr} R^2 - \frac{2}{3} (X_1^{(4)} + X_2^{(4)}) + 4 Y_{\bf{27}},\\
&=\frac{1}{2} \left(\frac{1}{2} K {\mathrm{tr}} R^2 +2 S {\mathrm{tr}}_{\bf{2}} F^2_1 + T {\mathrm{tr}}_{\bf{7}} F^2_2 \right)^2.
\end{split}
\end{align}
Hence, we  conclude that all the local anomalies are canceled via the Green--Schwarz--Sagnotti--West mechanism \cite{Green:1984bx,Sagnotti:1992qw}.

\section{Counting hypermultiplets: numerical oddities}\label{sec:odd}
The physics of D-branes in presence of singularities is full of subtleties \cite{AE1,CDE,Esole:2012tf}.  The situations  that are well understood rely on fundamental physical properties  such as  the cancellations of  anomalies and tadpoles or string dualities. 
For example, the induced D3-charge of a singular D7-brane is derived from a tadpole cancellation condition \cite{AE1, CDE}. 

The numbers of charged hypermultiplets we get from the triple intersection numbers of the five-dimensional theory match the numbers of hypermultiplets we get by solving the six-dimensional theory anomaly cancellation  conditions. 
The number of fundamental $n_{\bf{1,7}}$ can  be explained using Witten's genus formula and the number of bifundamental  n$_{\bf{2,7}}$ is given by an intersection number as expected by the usual D-brane picture. 

We point out an interesting observation about the number of hypermultiplets charged in the representation $(\bf{2,1})$ in the  \sug-model. 
One would expect that the number of fundamental matter $n_{\bf{2,1}}$ is also given by a direct intersection number, but that is not the case. 
We cannot explain the number $n_{\bf{2,1}}$  by looking at the D-brane configuration attached to this geometry.

\subsection{Determining  $n_{\bf{2,1}}$ from triple intersection numbers in $5d$ or anomalies in $6d$}
If we evaluate the number of multiplets in a five-dimensional supergravity theory with eight supercharges resulting from a compactification of  M-theory on an elliptically fibered Calabi-Yau threefold obtained by one of the resolution of the \sug-model, we can compare the triple intersection numbers of the divisors of the Calabi-Yau with the one-loop prepotential and deduce the number of multiplets. Moreover, in  a compactification of F-theory on the same variety, we get a six-dimensional supergravity theory and we can use anomaly cancellation conditions to determine the number of hypermultiplets. In both cases, we find that  $n_{\bf{2,1}}$ is given by 
\begin{equation}\nonumber
n_{\bf{2,1}}=-2S\cdot (4K+S+\frac{7}{4} T)=-2S\cdot (4K+S+2 T)+\frac{1}{2} S\cdot T.
\end{equation}

\subsection{Localization of hypermultiplets charged under the representation $({\bf{2,1}})$ }
 We recall that the discriminant of the \sug-model is supported on three divisors, namely, $S$ (which supports SU($2$)),  $T$ (which supports G$_2$) and $\Delta'=4 f^3 + 27 g^2 s$ (which does not support any gauge group).
While $S$  has a transverse intersection with $T$, $S$ intersects $\Delta'$ non-transversally along triple points located at $s=f=0$. 
The hypermultiplets transforming in the representation $(\bf{2,1})$ are localized at these  non-transverse intersection points of $S$ and $\Delta'$. 
The divisor $\Delta'$ has cuspidal singularities at $V(f,g)$. However, since  the divisor $S$ intersects it along its smooth locus,  the cuspdical points should not affect the discussion of the hypermultiplets charged under $(\bf{2,1})$.

The reduced locus of the intersection of $S$ and $\Delta'$ is $V(s,f)$. One might think that  following the usual D-brane picture,  the number of  hypermultiplets  $n_{\bf{2,1}}$ will just be the intersection number 
\begin{equation}
S\cdot [f]=-S\cdot (4K+S+2 T).
\end{equation}
  We might also think that since the intersection of $S$ and $\Delta'$ is given by the non-reduced scheme $(s,f^3)$,  we might take the multiplicity seriously and have 
  \begin{equation}
S\cdot[ \Delta']=-3S\cdot (4K+S+2 T).
\end{equation}
 But what we get is 
 \begin{equation} \label{Eq:n21}
   n_{\bf{2,1}}=-2S\cdot (4K +S+2T)+ \frac{1}{2}S\cdot T.
 \end{equation}
 As discussed in Remark \ref{rem1}, this formula reproduces what is known for the case when $T$ and $S$ are rational curves of self-intersection $-2$ and $-3$ intersecting transversally at one point \cite{Candelas:1997jz,Morrison:2012np}.
 
 \begin{rem}\label{rem.oddity}
There are two surprising numerical oddities about  this formula (equation \ref{Eq:n21}). 
\begin{enumerate}
\item The first oddity  is the coefficient of $2$  in the first term $-2S\cdot (4K +S+2T)$. 
\item The second numerical oddity is the presence of the second term as it predicts a nonzero  number of hypermultiplets charged in the representation $(\bf{2,1})$ even  when the fiber of type III does not intersect the component $\Delta'$ and therefore does not enhance to a fiber of type IV. 
\end{enumerate}
We will discuss these two points further below using the geometric point of view of this paper, namely, the fiber geometry in a crepant resolution and the notion of saturations of weights.
\end{rem}

\subsection{Mismatch between the brane and the anomaly picture?}
 We can justify the contribution $-2S\cdot (4K +S+2T)$ to $n_{\bf{2,1}}$  by analyzing  the crepant resolution. 
  To be specific, we consider the case of the Resolution I studied in Section \ref{Sec:IIII0SResI}.  The fiber III consists of two curves, a projective line ($C_0^\text{s}$) and a conic ($C_1^\text{s}$):
\begin{equation}
C_1^\text{s}:e_1=w_2 y^2- t^2s w_1 (fx+ gts)=0.
\end{equation}
The matrix of this conic is (with respect to the projective coordinates $[y:s:x]$):
\begin{equation}
M=
\begin{pmatrix}
w_2 & 0 & 0 \\
0  &0 &  - \frac{1}{2}w_1t^2 f\\
0  &  - \frac{1}{2}t^2 f w_1& - t^3 g w_1
\end{pmatrix}.
\end{equation}
The determinant of this matrix is 
\begin{equation}
\det M= -\frac{1}{4}w_2 w_1^2 t^4 f^2.\label{Eq.M}
\end{equation}
The determinant of $M$ vanishes over two distinct loci in the base: the intersection $T\cap S$ (in particular $V(e_1, w_1 w_2 t)$)  and  the intersection of $S$ and $V(f)$ (located on $S$ and away from $T$). 

At $f=0$, $M$ has  rank two and the conic splits into two lines inducing an enhancement 
\begin{equation}
\text{III}\to \text{IV}.
\end{equation}
Each of the  two line has weight $[1;0,0]$ since they are both away from the intersection with $T$ and each of them intersects $D_0^\text{s}$ transversally at one point (and therefore by linearity has intersection $-1$ with $D_1^\text{s}$).

The rank of the matrix $M$ collapses to two when $w_2=0$ and one when  $w_1=0$. We recall that after the three blowups, the divisor $T$ has total transform $w_1 w_2 t$. Over $S\cap T$, the generic fiber $C_1^\text{s}$ of $D_1^\text{s}$ degenerates to a collection of  four distinct rational curve  with multiplicities. We are at the intersection of $S$ and $T$ and the fiber III collides with the fiber I$_0^*$ to produce a non-Kodaira fiber corresponding to a Kodaira fiber of type III$^*$ (whose dual graph $\widetilde{\text{E}}_7$) with some components contracted to points. 
\begin{equation}
\text{III}+\text{I$_0^*$}\to\text{ III$^*$  with contracted components}. 
\end{equation}
Over $S\cap T$, the weights produced are in the  adjoint representation $(\bold{1,14})$ and the bifundamental representation $(\bold{2,7})$. 
When $f=0$, $C_1^\text{s}$ produces weights of the fundamental representation $(\bf{2,1})$ of \sug\ . 
One should keep in mind that the nonzero weights of the $(\bold{1,7})$ are the short nonzero weights of the adjoint representation.

The discriminant of the conic indicates that the number of hypermultiplets in the representation  $(\bf{2,1})$ of \sug\  at points away from $T$ might be  (taking into account the multiplicity of $f$ in $\det M$)
\begin{equation}\label{Eq.Vf}
S\cdot [\det M]=-2S\cdot (4K+S+2 T).
\end{equation}
In the expression of $n_{\bf{2,1}}$ in equation  \eqref{Eq:n21}, the  first term $-2S\cdot (4K+S+2 T)$ is exactly the contribution from the discriminant of the conic away from the divisor $T$ seen in equation \eqref{Eq.Vf}  when we take  into account the multiplicity of $f^2$ in $\det M$.  This explains the first numerical oddity mentioned in Remark \ref{rem.oddity}.

 \subsection{Hypers in the representation ($\bold{2,1}$)  at  $S\cap T$?}

The second term  $\tfrac{1}{2} S\cdot T$  in equation  \eqref{Eq:n21}  is a correction term that is equal to $n_{\bf{2,7}}$, the number of hypermultiplets  transforming in the bi-fundamental representation $(\bf{2,7})$. 
This correction term  gives an explicit mismatch with the number of hypermultiplets in the fundamental representation  $n_{\mathbf{2,1}}$ from what we expect from the D-brane picture which should be proportional to the intersection number computed in 
equation \eqref{Eq.Vf}.

One might argue that this second term  $\tfrac{1}{2} S\cdot T$ is an evidence of the existence of extra hypermultiplets transforming in the representation $(\bf{2,1})$ that is localized at the intersection of $S$ and $T$. This argument is in direct conflict with the fibral geometry of the \sug-model, since at the collision of $S$ and $T$, the geometry only predicts  bifundamental matters in the representation $(\bf{2,7})$.

\begin{rem} \label{last.rem}
The representation $\bf{2}$ is a minuscule representation of SU($2$) and $\bf{7}$ is a quasi-minuscule representation of G$_2$, in particular, the representation $\bf{7}$ of G$_2$ has a zero weight and all the six remaining nonzero weights are in the same orbit of the Weyl group.  
The representation $(\bf{2,7})$ is not a (quasi)-minuscule representation of \sug\   but its set of weights contains all the weights of the minuscule representation $(\bf{2,1})$  of \sug. 
 We note that the set of weights of the representation  $(\bf{2,1})$ is a maximal subweight system of the set of weights of the  representation  $(\bf{2,7})$. 
The weights of the  representation $(\bf{2,7})$ consists of two orbits under the Weyl group. One of the orbits has twelve elements that are the products of the weights of the $\bf{2}$ of SU$(2$) with the non-zero weights of the $\bf{7}$ of G$_2$. 
There also exists an orbit that is consisted of exactly of the  weights of the representation $(\bf{2,1})$ and is obtained by the direct product of the weights of the $\mathbf{2}$ of SU($2$) with the zero weight of the $\mathbf{7}$ of G$_2$.  
It follows that the weights of the fundamental representation $(\bf{2,1})$ are exactly the short weights of the bifundamental representation $(\bf{2,7})$. They form an orbit under the action of the Weyl group and the unique  sub-system of weights of the representation $(\bf{2,7})$. 
\end{rem}

In  the Resolution I, II, and III, the geometric weights in the representation $(\bf{2,7})$ are always weights of  the representation $(\bf{2,7})$  that are not weights of the representation $(\bf{2,1})$. In contrast, in the Resolution IV,  some of the rational curves that compose the singular fibers over  $S\cap T$  have weights  $[1;0,0]$, which is a weight of the representation $(\bf{2,1})$. However, on the same locus, we also get a rational curve with the weight $[-1;2,-1]$, which belongs to the representation $(\bf{2,7})$. 
As summarized on Table \ref{Table:Sat}, the saturation of these  two weights, $\{ [1;0,0],\ [-1;2,-1]\}$, is the full set of weights of the representation $(\bf{2,7})$ as can be seen directly  by computing their Weyl orbits (see Remark \ref{last.rem}). 
We note that the representation $(\bf{2,7})$ was derived in all the three other chambers. Hence, finding $(\bf{2,7})$ on $S\cap T$ in Resolution IV is expected since the representations seen in F-theory should not depend on a choice of a crepant resolution, as each crepant resolution corresponds to a distinct Coulomb chamber of the same gauge theory.  

  In conclusion, the number of SU($2$) fundamental hypermultiplets $n_{\bf{2,1}}$ is not given by an intersection number.  It has a correction that is at odd with the fibral geometry. While we see only the  representation $(\bf{2,7})$ at the intersection $S\cap T$, we note that the  representation $(\bf{2,1})$ is a minuscule representation of \sug\ and its weights form  a subset of the weights of the representation $(\bf{2,7})$. In that sense, the weights of the representation $(\bf{2,1})$ are visible in all chambers even though they appear explicitly only as weights of vertical irreducible curves in the Resolution IV. 
 It would be interesting to check if this phenomena occurs  in other models with semi-simple Lie groups.

\section*{Acknowledgements}
The authors are grateful to Patrick Jefferson, Ravi Jagadeesan,  and Shing-Tung Yau for helpful discussions. 
 The authors are thankful  to Jiahua Tian for comments on an earlier version of the paper. 
M.E. is supported in part by the National Science Foundation (NSF) grant DMS-1701635 ``Elliptic Fibrations and String Theory''.
M.J.K. would like to thank the University of Heidelberg and CERN for their hospitality and support during part of this work. 
M.J.K. is partly supported by the National Science Foundation (NSF) grant PHY-1352084.

\appendix

%%%%%%%%%%%%%%%%%%%%%%%%%%%
\bibliography{mboyoBib}

\end{document}